\documentclass[a4paper,UKenglish,cleveref, autoref, thm-restate]{lipics-v2021}

\usepackage{amsmath,amsthm,amssymb,appendix,bm,graphicx,hyperref,mathrsfs}
\newtheorem*{theorem*}{Theorem}

\usepackage{tikz}
\usetikzlibrary{positioning}
\usetikzlibrary {graphs}
\usetikzlibrary {chains,scopes}
\newcommand{\thup}[0]{\mathbin {\begin{tikzpicture}[scale=0.1]  \draw[->,thick] (0, 1.5 ) -- (2,1.5);     \draw[->,thick] (0,1.5) -- (0,3.5);	\end{tikzpicture} } }

\newcommand{\thdown}[0]{\mathbin{\begin{tikzpicture}[scale=0.1]  \draw[->,thick] (0, 1.5 ) -- (2,1.5);     \draw[->,thick] (0,1.5) -- (0,-0.5);	\end{tikzpicture} } }

\newcommand{\mitsuup}[0]{\forall 3\hspace{0.2em}\thup}

\newcommand{\mitsudown}[0]{\forall 3\hspace{0.2em}\thdown}

\newcommand{\eo}[0]{\textup{\textsf{EO}}}
\newcommand{\ceo}[0]{\textup{\#\textsf{EO}}}

\newcommand{\hol}[0]{\textup{\textsf{Holant}}}
\newcommand{\holodd}[0]{\textup{\textsf{Holant}}^\text{odd}}
\newcommand{\csp}[0]{\textup{\textsf{CSP}}}
\newcommand{\ccsp}[0]{\textup{\textsf{\#CSP}}}
\newcommand{\gc}[0]{\textsf{S}}
\newcommand{\len}[1]{\text{Len}(#1)} 
\newcommand{\T}[0]{\mathrm{T}}

\newcommand{\eoe}[0]{\text{HW}^=}
\newcommand{\eog}[0]{\text{HW}^\geq}
\newcommand{\eol}[0]{\text{HW}^\leq}
\newcommand{\eosg}[0]{\text{HW}^>}
\newcommand{\eosl}[0]{\text{HW}^<}

\newcommand{\su}[0]{\text{supp}}
\newcommand{\var}[0]{\text{Var}}
\newcommand{\ari}[0]{\text{arity}}

\newcommand{\pin}[0]{\Delta}

\newcommand{\eom}[1][\text{M}]{\textup{\textsf{EO}}^{#1}}

\newcommand{\ba}[1][0]{{{#1}-rebalancing}}
\newcommand{\pnp}[0]{\text{FP}^\text{NP}}

\newcommand{\sw}[0]{single-weighted}
\newcommand{\khol}[0]{\textup{\textsf{K-Holant}}}
\newcommand{\ii}[0]{\mathfrak{i}}

\title{From an odd arity signature to a Holant dichotomy} 

\author{Boning Meng\footnote{The authors share first author status.}}{Key Laboratory of System Software (Chinese Academy of Sciences) and State Key Laboratory of Computer Science, Institute of Software, Chinese Academy of Sciences; University of Chinese Academy of Sciences, Beijing 100080, China}{mengbn@ios.ac.cn}{https://orcid.org/0009-0006-0088-1639}{}
\author{Juqiu Wang\footnotemark[1]}{Key Laboratory of System Software (Chinese Academy of Sciences) and State Key Laboratory of Computer Science, Institute of Software, Chinese Academy of Sciences; University of Chinese Academy of Sciences, Beijing 100080, China}{wangjq21@ios.ac.cn}{https://orcid.org/0000-0001-9801-271X}{}
\author{Mingji Xia\footnotemark[1]}{Key Laboratory of System Software (Chinese Academy of Sciences) and State Key Laboratory of Computer Science, Institute of Software, Chinese Academy of Sciences; University of Chinese Academy of Sciences, Beijing 100080, China}{mingji@ios.ac.cn}{https://orcid.org/0000-0002-3868-9910}{}
\author{Jiayi Zheng\footnotemark[1]}{Key Laboratory of System Software (Chinese Academy of Sciences) and State Key Laboratory of Computer Science, Institute of Software, Chinese Academy of Sciences; University of Chinese Academy of Sciences, Beijing 100080, China}{zhengjy@ios.ac.cn}{https://orcid.org/0009-0005-5728-3616}{}
\authorrunning{B. Meng, J. Wang, M. Xia, J. Zheng}

\Copyright{B. Meng, J. Wang, M. Xia, J. Zheng} 


\relatedversion{} 

\funding{All authors are supported by National Key R\&D Program of China (2023YFA1009500), NSFC 61932002 and NSFC 62272448.}

\nolinenumbers 

\ArticleNo{}
\begin{document}
\maketitle

\begin{abstract}
   \textsf{Holant} is an essential framework in the field of counting complexity. For over fifteen years, researchers have been clarifying the complexity classification for complex-valued \textsf{Holant} on the Boolean domain, a challenge that remains unresolved.  In this article, we prove a complexity dichotomy for complex-valued \textsf{Holant} on Boolean domain when a non-trivial signature of odd arity exists. This dichotomy is based on the dichotomy for \textsf{\#EO}, and consequently is an $\text{FP}^\text{NP}$ vs. \#P dichotomy as well, stating that each problem is either in $\text{FP}^\text{NP}$ or \#P-hard. 

   Furthermore, we establish a generalized version of the decomposition lemma for complex-valued \textsf{Holant} on Boolean domain. It  asserts that each signature can be derived from its tensor product with other signatures, or conversely, the problem itself is in $\text{FP}^\text{NP}$. We believe that this result is a powerful method for building reductions in complex-valued \textsf{Holant}, as it is also employed as a pivotal technique in the proof of the aforementioned dichotomy in this article.
\end{abstract}

\section{Introduction}

Counting complexity is an essential aspect of computational complexity. Basically, it focuses on computing the sum of the weights of the solutions.
\hol\ is one of the most significant frameworks in the field of counting complexity, as it is capable of capturing a number of counting problems, such as counting perfect matchings in a graph or computing the partition function of the six-vertex model in statistical physics. There is a long history studying the complexity of \hol, and much progress has been made. On the other hand, however, the complete complexity classification for complex-valued \hol\ on Boolean domain remains unclear. 

 This article focuses on counting problems defined on Boolean domain and develops existing complexity dichotomies. It presents a full dichotomy for complex-valued \hol\ where there exists a non-trivial\footnote{By the term 'non-trivial' we mean that the signature does not remain constant at 0. It can be easily verified that if a signature remains constant at 0, then removing such signature from the signature set would not change the complexity of the problem.} signature of odd arity, which we denoted by complex $\holodd$ for short. This dichotomy encompasses the previous results in \cite{cai2020holantoddarity,backens2021full,yanghuangfu2024holant3}, with the exception that it is an $\pnp$ vs. \#P dichotomy. Furthermore, this article also presents a decomposition lemma for complex-valued \hol\ based on the dichotomy in \cite{meng2025fpnp}, which is the extension of a widely-used method originally developed in \cite{LinWang2018plushol}.

To begin with, $\hol$ is a framework of counting problems parameterized by a signature set $\mathcal{F}$, originally defined in \cite{CaiLuXia2009csp+hol}. Each signature of arity $r$ in $\mathcal{F}$ is a mapping from $\{0,1\}^r$ to $\mathbb{C}$, representing a local constraint. Each input, or equivalently an instance of $\hol(\mathcal{F})$, is a grid $\Omega=(G,\pi)$ where $G$ is a graph\footnote{In this article, the term 'graph' is understood to refer to a multi-graph. The presence of self-loops and parallel edges is invariably permitted.} and $\pi$ assigns a signature to each vertex in $G$. Each edge in $G$ is regarded as a variable subject to constraints imposed by incident vertices, and the output is the summation over the weights of all possible assignments of the variables, where the weight is the production of the output of each signature. See Definition \ref{def:Hol} for a detailed definition. 

$\hol$ is considered as one of the most important frameworks in the field of counting complexity. On one hand, it is capable of expressing a considerable number of counting problems, such as counting matchings and counting perfect matchings in a graph (\#\textsf{Matching} and \#\textsf{PM}), counting weighted Eulerian orientations in a graph (\ceo), counting constraint satisfaction problem (\ccsp) and computing the partition function of the six-vertex and eight-vertex model. On the other hand, \hol\ was originally inspired by the discovery of the holographic algorithm \cite{Valiant2006Accidental,valiant2008holographic}, which in turn was inspired by quantum computation and is closely related to the concept of Stochastic Local Operations and Classical Communication (SLOCC). Furthermore, several non-trivial connections between \hol\ and quantum computation have also been established \cite{cai2020holantoddarity,backens2017holant+,backens2021full,cai2024planar}.

Consequently, \hol\ has attracted a number of researchers to classify the complexity of this framework. The ultimate goal in this study is to decide the complexity of $\hol(\mathcal{F})$ for every complex-valued signature set $\mathcal{F}$. It is believed that there exists a dichotomy for complex-valued \hol, which means for arbitrary $\mathcal{F}$,  $\hol(\mathcal{F})$ is either polynomial-time computable or \#P-hard. Nevertheless, this hypothesis has been open for over 15 years, and currently the dichotomy result can only be stated for several types of signature sets. In the following, an overview of the research history is presented, along with a discussion of the current state of knowledge.

The research commences with symmetric signature sets, where the output of each signature only depends on the Hamming weight, or equivalently the number of 1's, of the input. Consequently, a symmetric signature $f$ of arity $r$ can be expressed as $[f_0, f_1, \ldots, f_r]_{r}$, where $f_i$ is the value of $f$ when the Hamming weight of the input is $i$. Under this restriction, the \hol\ framework is denoted as sym-\hol. Besides, sometimes the dichotomy holds when some specific signatures are available. 
We use $\pin_0,\pin_1,\pin_+,\pin_-$ to denote $[1,0],[0,1],[1,1],[1,-1]$ respectively. If all unary signatures are available, such problem is denoted as $\hol^*$; if $\pin_0,\pin_1,\pin_+,\pin_-$ are available, such problem is denoted as $\hol^+$; if $\pin_0,\pin_1$ are available, such problem is denoted as $\hol^c$. In particular, if a non-trivial signature of odd arity belongs to the signature set, we denote such problem as $\holodd$ for convenience. 

The previous results for symmetric cases are summarized in Table \ref{table:symhol}, and those for general cases are summarized in Table \ref{table:hol}. 
By definition, the result of each individual cell encompasses both the results that are located in the same column above it and the results that are located in the same row to its left. Furthermore, several essential dichotomies for $\hol$ not listed in the table are presented  below without exhaustive explanation. These are dichotomies respectively for \ccsp\ \cite{cai2014complexity}, six-vertex model \cite{CFX2018sixvertex}, eight-vertex model \cite{caifu2023eightvertex}, \ceo\ \cite{cai2020beyond,meng2024p,meng2025fpnp} and a single ternary signature \cite{yanghuangfu2024holant3}.

\begin{table}[]
\begin{center}
\begin{tabular}{|l|l|l|l|}
\hline
             & sym-\hol$^*$                          & sym-\hol$^c$                          & sym-\hol                              \\ \hline
real-valued & /                                       & \cite{CaiLuXia2009csp+hol} & \cite{HL2016realsymHolant} \\ \hline
complex-valued & \cite{CaiLuXia2009csp+hol} & \cite{CHL2012symholantc}   & \cite{CGW2013symHolant}    \\ \hline
\end{tabular}
\caption{Research history of sym-\hol. The term '/' denotes that the corresponding dichotomy is directly encompassed by another dichotomy, thus obviating the necessity for separate research.}
\label{table:symhol}
\begin{tabular}{|l|l|l|l|l|l|}
\hline
                                               & \hol$^*$                         & \hol$^+$                         & \hol$^c$                             & $\hol^\text{odd}$ & \hol                                    \\ \hline
nonnegative-valued & /                                                                      & /                                  & /                                      & /          & \cite{LinWang2018plushol}    \\ \hline
real-valued         & /                                  & /                                  & \cite{cai2018realholantc} &   \cite{cai2020holantoddarity}     & \cite{shao2020realholant} \\ \hline
complex-valued          & \cite{CLX2011Holant*} & \cite{backens2017holant+} & \cite{backens2021full}    &  \textbf{Our result}         & Unknown                                        \\ \hline
\end{tabular}
 \caption{Research history of \hol. The term '/' is same as that in Table \ref{table:symhol}.}
 \label{table:hol}
\end{center}
\end{table}

It is noticeable that  except for the last column in Table \ref{table:hol}, the majority of works concentrate on scenarios where specific signatures of odd arity are present. Furthermore, the number of such signatures decreases from left to right, being infinite, 4, 2, 1 (and 0 in the last column) respectively. This article proves a definitive conclusion to this series of works, the position of which is also presented in Table \ref{table:hol}.

In order to reach this conclusion, it is necessary to extend the result of the decomposition lemma, originally developed in \cite{LinWang2018plushol}. This lemma can derive a specific signature from its tensor product with other signatures, and has been employed in numerous recent works \cite{cai2020beyond,cai2020holantoddarity,shao2020realholant,caifu2023eightvertex,meng2025fpnp}. However, the general form of this lemma does not hold for all complex-valued signature sets.

On the other hand, a series of works has classified the complexity of counting weighted Eulerian orientations in a graph (\ceo) \cite{CFX2018sixvertex,cai2020beyond,shao2024eulerian,meng2024p,meng2025fpnp}. An  $\pnp$ vs. \#P dichotomy dichotomy is established in \cite{meng2025fpnp}, which claims that a problem in \ceo\ is either in $\pnp$ or \#P-hard.
This motivates us to prove a generalized version of the decomposition lemma, stating that either the signature can be derived from the tensor product, or the complexity of the problem can be classified by the \ceo\ dichotomy.

Consequently, there are two major results in this article, stated as follows. The detailed forms of them are Theorem \ref{thm:decomposition lemma} and Theorem \ref{thm:main theorem}, presented in Section \ref{main}. We use $ \leq_T$ to denote polynomial-time Turing reduction and $\equiv_T$ to denote polynomial-time Turing equivalence. We use $\otimes$ to denote tensor product.

\begin{theorem}[Decomposition lemma]\label{thm:easydecompo}
    Let $\mathcal{F}$ be a set of signatures and $f$, $g$ be two signatures. Then one of the following holds.
\begin{enumerate}
    \item $\hol(\mathcal{F},f,g)\equiv_T\hol(\mathcal{F},f\otimes g)$;
    \item $\hol(\mathcal{F},f\otimes g)$ is in $\pnp$.
\end{enumerate}
\end{theorem}

\begin{theorem}[Dichotomy for $\holodd$]\label{thm:easymain}
    Let $\mathcal{F}$ be a set of signatures that contains a non-trivial signature of odd arity. Then one of the following holds.
\begin{enumerate}
    \item $\hol(\mathcal{F})$ is \#P-hard;
    \item $\hol(\mathcal{F})$ is in $\pnp$.
\end{enumerate}
\end{theorem}

In some cases appearing in Theorem \ref{thm:easydecompo}, $\hol(\mathcal{F},f\otimes g)$ is \#P-hard. In these cases, the reduction in the first statement trivially holds. We also remark that the definitions of complexity classes appearing in the two theorems can be found in \cite{arora2009computational}. In particular, the complexity class $\pnp$ is introduced due to the dichotomy for \ceo\ in \cite{meng2025fpnp}, making Theorem \ref{thm:easymain} an $\pnp$ vs. \#P dichotomy. As analyzed in \cite{meng2025fpnp}, such dichotomy can separate the computational complexity as well, unless the complexity class PH collapses at the second level. Furthermore, as shown in the proof strategy in Section \ref{main}, the proof of Theorem \ref{thm:easymain} is unfeasible in the absence of the dichotomy for  \ceo, hence the introduction of the complexity class $\pnp$ is somehow inevitable. On the other hand, when the intermediate results are independent of the dichotomy for  \ceo,  the $\pnp$ term does not manifest. In other words, if a traditional FP vs. \#P dichotomy is later proved for \ceo, Theorem \ref{thm:easymain} will consequently become an FP vs. \#P dichotomy as well.

Despite the $\pnp$ part, this article makes a great progression towards the complexity classification for complex-valued \hol. With Theorem \ref{thm:easymain}, it is now sufficient for future studies to focus on the case when all signatures in $\mathcal{F}$ are of even arity. With Theorem \ref{thm:easydecompo}, we may further assume that these signatures are irreducible, and each realized signature does not contain a factor of odd arity. 
In addition, Theorem \ref{thm:easymain} also encompasses the preceding results in \cite{cai2020holantoddarity,backens2021full,yanghuangfu2024holant3}. To the best of our knowledge, the dichotomy results in the field of complex-valued \hol\ that have not been yet encompassed are precisely those in \cite{cai2013vanishing,shao2020realholant,caifu2023eightvertex,meng2025fpnp} and this article. 

This article is organized as follows.
In Section \ref{preli}, we introduce preliminaries needed in this article.
In Section \ref{main}, we present our main results in detail, and explain our proof strategy.
In Section \ref{sec:decomposition lemma}, we prove Theorem \ref{thm:easydecompo}.
In Section \ref{oddsi}, we prove Theorem \ref{thm:easymain}.
In Section \ref{ccls}, we conclude our results.

\section{Preliminaries}\label{preli}

\subsection{Definitions and notations}

A \textit{Boolean domain} refers to the set $\{0,1\}$, or $\mathbb{F}_2$ for convenience to describe particular classes of tractable signatures. A \textit{complex-valued Boolean signature} $f$ with $r$ variables is a mapping from $\{0,1\}^r$ to $\mathbb{C}$. In particular, this article focuses on algebraic complex-valued Boolean signatures. Given a 01-string $\alpha$ with length $r$, we use $f(\alpha)$ or $f_{\alpha}$ to denote the value of $f$ on the input $\alpha$. The set of variables of $f$ is denoted by $\var(f)$, and its size (or arity) is denoted by $\ari(f)$. The set of strings on which $f$ has non-zero values is the \textit{support} of $f$, denoted by $\su(f)$.

Let $\alpha$ denote a binary string. For $s \in \{0,1\}$, we use $\#_s(\alpha)$ to denote the number of occurrences of $s$ in $\alpha$, and $\#_1(\alpha)$ is also known as the \textit{Hamming weight} of $\alpha$. The length of $\alpha$ is denoted by $\len\alpha$, and $\alpha_i$ refers to its $i$-th bit. We use $0_k$ and $1_k$ to denote the two strings of length $k$ that only contain 0 and 1 respectively. When the length $k$ is clear from the context, we use $\textbf{0}$ and $\textbf{1}$ to represent $0_k$ and $1_k$.

The following notations from \cite{meng2024p,meng2025fpnp} are employed in this article.

$$
\eoe = \{\alpha \mid \#_1(\alpha) = \#_0(\alpha)\} \quad \text{and} \quad \eog = \{\alpha \mid \#_1(\alpha) \geq \#_0(\alpha)\}.  
$$  
A signature $f$ is an $\eoe$ signature or equivalently an $\eo$ signature if $\su(f) \subseteq \eoe$, and such signatures inherently have even arity. Analogously, $f$ is an $\eog$ signature if $\su(f) \subseteq \eog$, and $\mathcal{F}$ is $\eog$ if $\mathcal{F}$ only contains $\eog$ signatures. Similar notations are also defined by replacing "$\geq$" with "$\leq$", "$<$", or "$>$". 

The \textit{bitwise addition} (XOR) of two strings $\alpha$ and $\beta$, denoted by $\alpha \oplus \beta$, produces a string $\gamma$ where $\gamma_i = \alpha_i + \beta_i \pmod 2$ for each $i$. A set $A$ of strings is \textit{affine} if $\alpha \oplus \beta \oplus \gamma \in A$ for arbitrary $\alpha, \beta, \gamma \in A$. The \textit{affine span} of a set $S$ is the minimal affine space containing $S$. For a signature $f$, $\text{Span}(f)$ denotes the affine span of its support.

A \textit{symmetric signature} $f$ of arity $r$ is expressed as $[f_0, f_1, \ldots, f_r]_{r}$, where $f_i$ is the value of $f$ on all inputs of Hamming weight $i$ and the subscript $r$ can be omitted without causing ambiguity. 
Commonly used examples include the unary signatures $\Delta_0 = [1, 0]$ and $\Delta_1 = [0, 1]$, and the binary disequality signature $(\neq_2) = [0, 1, 0]$. 

Let $\mathcal{EQ}$ denote the class of \textit{equality signatures}, defined as $\mathcal{EQ} = \{=_1, =_2, \ldots, =_r, \ldots\}$, where $=_r$ is the signature $[1, 0, \ldots, 0, 1]_r$ of arity $r$. In other words, the output of $=_r$ is 1 if all bits of the input are identical (all 0’s or all 1’s), and is 0 otherwise.  

The tensor product of signatures is denoted by $\otimes$. We use $\langle\mathcal{F}\rangle$ to denote the closure of $\mathcal{F}$ under tensor products.

\subsection{Frameworks of counting problems}

In this subsection, we present several pivotal frameworks of counting problems, including $\hol$, $\ceo$ and $\ccsp$. We largely refer to \cite[Section 1.2]{cai2017complexity} and \cite[Section 2.1]{cai2020beyond}.

Let $\mathcal{F}$ denote a fixed finite set of Boolean signatures. A signature grid over $\mathcal{F}$, denoted as $\Omega(G, \pi)$, consists of a graph $G=(V,E)$ and a mapping $\pi$ that assigns to each vertex $v \in V$ a signature $f_v \in \mathcal{F}$, along with a fixed linear order of its incident edges. The arity of each signature $f_v$ matches the degree of $v$, with each incident edge corresponding to a variable of $f_v$. Throughout, we allow graphs to have parallel edges and self-loops. Given any 0-1 edge assignments $\sigma$, the evaluation of the signature grid is given by the product $\prod_{v \in V} f_v(\sigma|_{E(v)})$, where $\sigma|_{E(v)}$ represents the restriction of $\sigma$ to the edges incident to $v$. 

\begin{definition}[$\hol$ problems]
    A $\hol$ problem, $\hol(\mathcal{F})$, parameterized by a set $ \mathcal{F} $ of complex-valued signatures, is defined as follows: Given an instance $I$, that is a signature grid $ \Omega(G, \pi) $ over $ \mathcal{F} $, the output is the partition function of $ \Omega $,
    
    $$
    \text{Z}(I)=\hol_\Omega = \sum_{\sigma:E\rightarrow\{0,1\}} \prod_{v \in V} f_v(\sigma|_{E(v)}).
    $$
The bipartite $\hol$ problem $\hol(\mathcal{F} \mid \mathcal{G})$ is a specific case where $ G $ is bipartite, say $ G(U, V, E) $, with vertices in $ U $ assigned signatures from $ \mathcal{F} $ and vertices in $ V $ assigned signatures from $ \mathcal{G}$. We denote the left-hand side (or the right-hand side) in bipartite \hol\ briefly as LHS (or RHS).
\label{def:Hol}
\end{definition}

For simplicity of presentation, we omit the curly braces when writing the set of one single signature. For example, we write $\hol(\mathcal{F}\cup\{f\})$ as $\hol(\mathcal{F},f)$. In addition to bipartite \hol, there are several variants of $\hol$. We use $\hol^c(\mathcal{F})$ to denote $\hol(\mathcal{F},\Delta_0,\Delta_1)$. We use $\khol(\mathcal{F})$ to denote $\hol(\neq_2\mid\mathcal{F})$ and $\khol^c(\mathcal{F})$ to denote $\khol(\mathcal{F},\Delta_0,\pin_1)$. 

Another special variant of \hol\ is counting weighted Eulerian Orientations (\ceo). Let $\mathcal{F}$ denote a fixed finite set of $\eo$ signatures. An $\eo$-signature grid, denoted by $\Omega(G, \pi)$, is defined on an Eulerian graph $G = (V, E)$, where every vertex $v \in V$ has positive even degree. The definition of $\pi$ is similar to that of a signature grid in $\hol$ problems. For an Eulerian graph $G$, let $\eo(G)$ denote the set of all Eulerian orientations of $G$. In such an orientation, every edge is assigned a direction. For each edge, the head is assigned $0$ and the tail is assigned $1$ to represent the direction. This assignment ensures that each vertex $v$ has an equal number of $0$'s and $1$'s incident to it. The weight of a vertex $v$ under an orientation $\sigma \in \eo(G)$ is determined by $f_v$ evaluated on the local assignment $\sigma|_{E(v)}$, which restricts $\sigma$ to the edges incident to $v$. The global evaluation of $\sigma$ is the product $\prod_{v \in V} f_v(\sigma|_{E(v)})$. 

\begin{definition}[$\#\eo$ problems {\cite{cai2020beyond}}]\label{def:numbereo}  
Given a set $\mathcal{F}$ of $\eo$ signatures, the $\#\eo$ problem $\#\eo(\mathcal{F})$ is defined as follows: For an instance $I$, that is an $\eo$-signature grid $\Omega(G, \pi)$ over $\mathcal{F}$, compute the partition function:  

$$
\text{Z}(I)=\#\eo_\Omega = \sum_{\sigma \in \eo(G)} \prod_{v \in V} f_v(\sigma|_{E(v)}).
$$  
\end{definition}

This framework generalizes classical models in statistical physics. For instance, the six-vertex model corresponds to a $\ceo$ problem defined by a single quaternary $\eo$ signature.

In this work, we apply polynomial-time Turing reductions (denoted by $ \leq_T$) and equivalences (denoted by $ \equiv_T $) to derive complexity classifications. A fundamental result is the representation of $\#\eo$ problems as a restricted subclass of bipartite $\hol$ problems. 

\begin{lemma}[\cite{cai2020beyond}]  
For any finite set $\mathcal{F}$ of $\eo$ signatures, $\#\eo(\mathcal{F}) \equiv_T \hol(\neq_2 \mid \mathcal{F}).$
\end{lemma}  

It is noteworthy that $\#\eo(\mathcal{F})$ differs from $\khol(\mathcal{F})$ that $\ceo$ is parameterized by an $\eo$ signature set. 

The framework of \ccsp\ is also closely related to \hol\ \cite{CaiLuXia2009csp+hol}. A variant of \ccsp, denoted by $\ccsp_d$ problems \cite{HL2016realsymHolant}, also plays a role in our proof.

\begin{definition}[$\#\csp$]\label{def:csp}  
Let $\mathcal{F}$ be a fixed finite set of complex-valued signatures over the Boolean domain. An instance of $\#\csp(\mathcal{F})$ consists of a finite variable set $X = \{x_1, x_2, \ldots, x_n\}$ and a clause set $C$, where each clause $C_i \in C$ includes a signature $f_i=f \in \mathcal{F}$ (of arity $k$) and a selection of variables $(x_{i_1}, x_{i_2}, \ldots, x_{i_k})$ from $X$, allowing repetition. The output of the instance is defined as:  

\begin{definition}\label{def:cspd} 
    Let $d\geq1$ be an integer and let $\mathcal{F}$ be a set of complex-valued signatures. The problem $\ccsp_d(\mathcal{F})$ is the restriction of $\ccsp(\mathcal{F})$ to the instances where every variable occurs a multiple of $d$ times.
\end{definition}
$$
\text{Z}(I)=\sum_{\substack{x_1, \ldots, x_n \in \{0,1\}}} \prod_{(f, x_{i_1}, \ldots, x_{i_k}) \in C} f(x_{i_1}, \ldots, x_{i_k}),  
$$    
\end{definition}  

For an integer $d\geq1$, let $\mathcal{EQ}_d$ denote the class of equality signatures whose arities are divisible by $d$, defined by $\mathcal{EQ}_d = \{=_d, =_{2d}, \ldots, =_{rd}, \ldots\}$. It is well-known that $\ccsp$ and $\ccsp_d$ can be expressed within the $\hol$ framework, as shown in the following lemmas.  

\begin{lemma}[{\cite[Lemma 1.2]{cai2017complexity}}]\label{lemma:csp-holant}  
 $\#\csp(\mathcal{F}) \equiv_T \hol(\mathcal{EQ} \cup \mathcal{F})$ 
\end{lemma}  

\begin{lemma} 
 For an integer $d\geq1$, $\ccsp_d(\mathcal{F}) \equiv_T \hol(\mathcal{EQ}_d \mid \mathcal{F}).$ 
\end{lemma}

\subsection{Fundamental methods}

In this subsection, we present three pivotal reduction techniques in the study of counting problems: gadget construction, polynomial interpolation, and holographic transformation. Additionally, we introduce the theory of unique tensor decomposition for signatures. These methods constitute essential tools required for the proofs developed in this article.

\subsubsection{Gadget construction and signature matrix}

This subsection introduces two principal concepts in the study of counting problems: \textit{gadget construction} and the \textit{signature matrix}. Gadget construction is a critical reduction technique, and the signature matrix provides an algebraic representation bridging gadget construction with matrix multiplication.

Let $\mathcal{F}$ denote a set of signatures. An $\mathcal{F}$-gate is similar to a signature grid $\Omega(G,\pi)$, while the edges of $G = (V, E, D)$ are partitioned into internal edges $E$ (edges that have two ends from $V$) and dangling edges $D$ (edges that are incident to a single vertex). Let $|E| = n$ and $|D| = m$, with internal edges encoding variables $\{x_1, \ldots, x_n\}$ and dangling edges corresponding to $\{y_1, \ldots, y_m\}$. The $\mathcal{F}$-gate induces a signature $f: \{0,1\}^m \to \mathbb{C}$ defined by:

$$
f(y_1, \ldots, y_m) = \sum_{\sigma: E \to \{0,1\}} \prod_{v \in V} f_v\left(\hat{\sigma}|_{E(v)}\right),
$$
where $\hat{\sigma}: E \cup D \to \{0,1\}$ extends $\sigma$ by incorporating the assignment $\mathbf{y} \in \{0,1\}^m$ to dangling edges, and $f_v$ denotes the signature assigned to vertex $v$ via $\pi$. A signature $f$ is called realizable from $\mathcal{F}$ if it can be represented by an $\mathcal{F}$-gate. If $E = \emptyset$, the resulting signature is the tensor product $\bigotimes_{v \in V} f_v$. We denote the set of all signatures realizable from $\mathcal{F}$ as $\gc(\mathcal{F})$.
For any signature set $\mathcal{F}$ and $f\in\gc(\mathcal{F})$, it is shown in \cite[Lemma 1.3]{cai2017complexity} that  
$
\hol(\mathcal{F}) \equiv_T \hol(\mathcal{F},f)
$  
, which provides a powerful tool to build reductions.

Next we introduce the signature matrix, which is an algebraic representation of a given signature. This representation is conducive to computational convenience.
For a signature $f: \{0,1\}^r \to \mathbb{C}$ with ordered variables $(x_1,\ldots,x_r)$, its signature matrix with parameter $l$ is $M(f) \in \mathbb{C}^{2^l \times 2^{r-l}}$. Its row indices correspond to the assignments to the first $l$ variables $(x_1, ..., x_l)$ and its column indices correspond to the assignments to the remaining $r-l$ variables $(x_{l+1}, ..., x_r)$. The entry is the corresponding value of $f$ given the input that combines the row index and the column index. For even arity signatures, $l = r/2$ is commonly used. The matrix is denoted by $M_{x_1\cdots x_l,x_{l+1}\cdots x_r}(f)$ or simply $M_f$ without causing ambiguity.

An example is a binary signature $f=(f_{00},f_{01},f_{10},f_{11})$, and it has a matrix form as:

   $$
   M_f = \begin{pmatrix} 
   f_{00} & f_{01} \\ 
   f_{10} & f_{11} 
   \end{pmatrix}.
   $$  
   
We now present the relation between matrix multiplication and gadget constructions. Let $f,g \in \gc(\mathcal{F})$ are of arity $n$ and $m$ respectively. Suppose $\var(f)=\{x_1,\ldots,x_n\}$ and $\var(g)=\{y_1,\ldots,y_m\}$. Connecting $l$ dangling edges $\{x_{n-l+1},...,x_n\}$ of $f$ to dangling edges $\{y_1,\ldots,y_l\}$ of $g$ yields a resulting signature $h \in \gc(\mathcal{F})$ with arity $n+m-2l$ and $\var(h)=\{x_1,\ldots,x_{n-l},y_{l+1},\ldots,y_m\}$. We have

$$
M_h = M_{x_1 \ldots x_{n - l}, y_{l + 1} \ldots y_m} = M_{x_1 \ldots x_{n - l}, x_{n - l + 1} \ldots x_n} \cdot M_{y_1 \ldots y_l, y_{l + 1} \ldots y_m} = M_f \cdot M_g.
$$  

We also address several typical operations in gadget construction. The first operation is \textit{adding a self-loop}. Suppose $f$ is a signature and $\var(f)=\{x_1,\ldots,x_n\}$. Suppose $b$ is a binary signature with $M_b = \begin{pmatrix} 
   b_{00} & b_{01} \\ 
   b_{10} & b_{11} 
   \end{pmatrix}.$ By adding a self-loop using a binary signature $b$ on $x_1$ and $x_2$ of $f$ we mean connecting $x_1$ to the first variable of $b$ and  $x_2$ to the second variable of $b$, obtaining a signature $f'$. Here we use $\textbf{x}$ to denote $(x_3,\ldots,x_n)$ and we have:

$$f'(\textbf{x})=b_{00}f(0,0,\textbf{x})+b_{01}f(0,1,\textbf{x})+b_{10}f(1,0,\textbf{x})+b_{11}f(1,1,\textbf{x}).$$
We use $\partial_{ij}f$ (or $\widehat{\partial_{ij}}\widehat{f}$) to denote the signature obtained by adding a self-loop using $=_2$ (or $\neq_2$) on $x_i$ and $x_j$ of $f$ (or $\widehat{f}$).

The second operation is \textit{pinning}. Given $f$ as a signature, the pinning operation connects $\pin_0$ or $\pin_1$ to one of $f$'s variables. We use $f^{x_1=0}$ and $f^{x_1=1}$ to denote the resulting signature after connecting the first variable $x_1$ to $\pin_0$ and $\pin_1$ respectively.

The following lemma will be used in our proof.

\begin{lemma}[{\cite[Lemma 3.9]{cai2020holantoddarity}}]\label{lem:shaoshuai lemma 3.9}
    Let $f$ be a signature of arity $n\geq2$. If for any index $i$, by pinning the variable $x_i$ of $f$ to $0$, we have $f^{x_i=0}\equiv0$, then $f_\alpha=0$ for any $\alpha$ satisfying $\#_1(\alpha)\neq n$. Furthermore, if there is a pair of indices $\{j,k\}$ such that $\partial_{jk}f\equiv0$, then $f\equiv0$.
\end{lemma}

\subsubsection{Holographic transformations and SLOCC}

In the study of $\hol$ problems, holographic transformations provide a powerful tool for complexity classification. This section formalizes key concepts and theorems, and presents their relation with quantum computation.
 
For any graph $ G $, a bipartite graph preserving the $\hol$ value can be constructed via the \textit{2-stretch} operation: each edge $ e $ is replaced by a path of length two, introducing a new vertex assigned $ =_2 $. We present the equivalence that $\hol(=_2 \mid \mathcal{F}) \equiv_T\hol(\mathcal{F}) $, where $ \mathcal{F} $ denotes a set of signatures.
 
Let $ T \in \mathbf{GL}_2(\mathbb{C}) $ be an invertible 2×2 matrix. For a signature $ f $ of arity $ n $, represented as a column vector $ f \in \mathbb{C}^{2^n} $, the transformed signature is denoted by $ Tf = T^{\otimes n}f $. For a signature set $\mathcal{F}$, we define $T\mathcal{F}=\{Tf\mid f\in\mathcal{F}\}$. The transformation of contra-variant signatures (row vectors) is $ fT^{-1} $. When we write $Tf$ or $T\mathcal{F}$, $f$ and signatures in $\mathcal{F}$ are regarded as column vectors by default; similarly for $fT$ or $\mathcal{F}T$ as row vectors. We also use $T f$ to denote the matrix $M_{V,\var(f)-V}(T f)$, where $V$ is a subset of $\var(f)$. Suppose $|V|=k$, we have $M_{V,\var(f)-V}(T f)=T^{\otimes k} M_{V,\var(f)-V}(f)\left(T^{\mathrm{T}}\right)^{\otimes (n-k)}$. Similarly, $M_{V,\var(f)-V}\left(f T^{-1}\right)=\left(T^{-\T}\right)^{\otimes k} M_{V,\var(f)-V}(f)\left(T^{-1}\right)^{\otimes (n-k)}$, where we briefly denote $(T^{-1})^{\T}$ as $T^{-\T}$ and $|V|=k$.

Suppose $T\in\textbf{GL}_2(\mathbb{C})$. A holographic transformation defined by $ T $ applies $T$ to all signatures in the RHS and $ T^{-1} $ to all signatures in the LHS. In other words, given a signature grid $ \Omega = (H, \pi) $ of $\hol(\mathcal{F} \mid \mathcal{G}) $, the transformed grid $ \Omega' = (H, \pi') $ is an instance of $ \hol(\mathcal{F}T^{-1} \mid T\mathcal{G}) $.

\begin{theorem}[{\cite{valiant2008holographic}}]\label{thm:holographic transformation equivalence}
For any $ T \in \mathbf{GL}_2(\mathbb{C}) $,  

$$
\hol(\mathcal{F} \mid \mathcal{G}) \equiv_T \hol(\mathcal{F}T^{-1} \mid T\mathcal{G}).
$$ 
\end{theorem}
Theorem \ref{thm:holographic transformation equivalence} shows that holographic transformations preserve the complexity of bipartite $\hol$ problems. We use $\mathscr{O}$ to denote all orthogonal matrices, that is $\mathscr{O}=\{O\in\textbf{GL}_2(\mathbb{C})\mid O^{\T}O=\begin{pmatrix}
    1 & 0 \\
    0 & 1
\end{pmatrix}\}.$ Notably, any holographic transformation by $ O \in \mathscr{O} $ keeps the binary Equality signature invariant: $ (=_2)O^{-1} = (=_2) $, hence $ \hol(=_2 \mid \mathcal{F}) \equiv_T \hol(=_2 \mid O\mathcal{F}) $.

A pivotal holographic transformation employs $ K^{-1} = \frac{1}{\sqrt{2}}\begin{bmatrix} 1 & -i \\ 1 & i \end{bmatrix} $, mapping $ (=_2) $ to the disequality signature $ (=_2)K=(\neq_2) $, yielding $ \hol(=_2 \mid \mathcal{F}) \equiv_T \hol(\neq_2 \mid K^{-1}\mathcal{F}) $.  We use $\widehat{f}$ and $\widehat{\mathcal{F}}$ to denote $K^{-1}f$ and $K^{-1}\mathcal{F}$ respectively, and $X=\begin{pmatrix}
  0 & 1 \\
  1 & 0\end{pmatrix}$ to denote the signature matrix of $\neq_2$. It can be verified that $K=\frac{1}{\sqrt{2}}\left(\begin{matrix}
  1 & 1 \\
  \mathfrak{i} & -\mathfrak{i}\end{matrix}\right)
$. Furthermore, with these notations, we have $\hol(\mathcal{F}) \equiv_T \khol(\widehat{\mathcal{F}}) $.

We are also interested in what kind of holographic transformations preserves $\neq_2$ invariant, and by direct computation we have the following lemma.
\begin{lemma}
    Let $\widehat{Q}=\begin{pmatrix}
        1/q & 0\\
        0 & q
    \end{pmatrix}$ or $\begin{pmatrix}
        0 & 1/q \\
        q & 0
    \end{pmatrix}$, where $q\neq 0$. Then  $(\neq_2)\widehat{Q}=(\neq_2)$ and $\khol(\mathcal{F})\equiv_T\khol(\widehat{Q}^{-1}\mathcal{F})$. Furthermore, any matrix $A$ satisfying $(\neq_2)A=(\neq_2)$ is of the form $\begin{pmatrix}
        1/q & 0\\
        0 & q
    \end{pmatrix}$ or $\begin{pmatrix}
        0 & 1/q \\
        q & 0
    \end{pmatrix}$, where $q\neq0$.
    \label{lem:quanxi}
\end{lemma}

Sometimes, it is more convenient to study a \hol\ problem in the setting of \khol. As presented in Lemma \ref{lem:quanxi}, holographic transformations that preserve $\neq_2$ invariant do not change the support of signatures on the RHS, except for a possible exchange of the symbol of $0$ and $1$ for all signatures. Consequently, many results are proved by analyzing the support in the setting of \khol. The following lemma is such an example, which is also useful in our proof.
\begin{lemma}[\cite{cai2020beyond}]\label{lem:wangjuqiuwangle}
    Let $\widehat{f}$ be a signature of arity $k\ge 3$. If for any indices $1 \le i,j\le k$, by adding a self-loop on $x_i$ and $x_j$ of $\widehat{f}$ using $\neq_2$, $\widehat{\partial_{ij}}\widehat{f}\equiv 0$, then $\widehat{f}(\alpha)=0$  for arbitrary $\alpha$ with $0 < \#_1(\alpha) < n$.
\end{lemma}

A concept in the quantum computation theory called Stochastic Local Operations and Classical Communication (SLOCC) \cite{bennett2000exact,dur2000three} is a generalization of the holographic transformation, which allows distinct transformations on each qubits, or variables in the setting of \hol. For an $ n $-ary signature $ f $, SLOCC applies $ M_1 \otimes \cdots \otimes M_n $ with $ M_i \in \mathbf{GL}_2(\mathbb{C}) $ on $f$. When $ M_i = T $ for each $1\leq i\leq n$, it is exactly the holographic transformation defined by $T$. Backens \cite{backens2017holant+} leveraged SLOCC to establish the dichotomy theorem for $ \hol^+$  and $\hol^c$, which implies relations between counting problems and the quantum computation theory.

Besides, a classification of ternary irreducible signatures under SLOCC is given in \cite{dur2000three}.

\begin{lemma}[{\cite{dur2000three}}]
    Suppose $g$ is a ternary irreducible signature. Then one of the following holds.
    \begin{enumerate}
        \item $g$ is of GHZ type. That is, $g=(M_1\otimes M_2\otimes M_3)[1,0,0,1]$;
        \item $g$ is of W type. That is, $g=(M_1\otimes M_2\otimes M_3)[0,1,0,0]$,
    \end{enumerate}
    where $M_1,M_2,M_3\in\mathbf{GL}_2(\mathbb{C})$.
\end{lemma}
\subsubsection{Polynomial interpolation}
Polynomial interpolation is a powerful tool to build reductions as well. Detailed information about it can be found in \cite{cai2017complexity}, and it is sufficient for us to use the following lemma in our proof, whose proof follows from polynomial interpolation.

\begin{lemma}\label{lem:interpolation}
    In the setting of $\khol(\widehat{\mathcal{F}})$, if $(n+1)$ pairwise linearly independent unary signatures with $poly(n)$-size can be realized in poly(n) time for any $n\in\mathbb{N}_+$ in the setting of \khol, then $\khol(\widehat{\mathcal{F}},[1,-1])\leq_T\khol(\widehat{\mathcal{F}})$.
\end{lemma}

\subsubsection{Signature factorization}
In the study of $\hol$ problems, the algebraic structure of signatures plays a pivotal role. This section introduces the \textit{Unique Prime Factorization} (UPF), which is a fundamental method of characterizing signatures.
 
A non-trivial signature $ f $ is irreducible if it cannot be expressed as a tensor product $ f = g \otimes h $ for non-constant signatures $ g, h $. Then we introduce the Unique Prime Factorization.

\begin{lemma}[{\cite[Lemma 2.13]{cai2020beyond}}]\label{lem:upf}
Every non-trivial signature $ f $ has a prime factorization $ f = g_1 \otimes \cdots \otimes g_k $, where each $ g_i $ is irreducible, and the factorization is unique up to variable permutation and constant factors. That is to say, if $ f = g_1 \otimes \cdots \otimes g_k = h_1 \otimes \cdots \otimes h_\ell $, then $ k = \ell $, and there exists a permutation $\pi$ such that $ g_{\pi(i)}=c_i h_i $ for each $1\leq i\leq k $, where each $c_i$ is a constant.
\end{lemma}

\subsection{Known dichotomies}
In this section, we introduce important results of complexity classifications for counting problems. 

\subsubsection{\ccsp\ and $\ccsp_d$}
Let $ X = (x_1, x_2, \ldots, x_d, 1)^\T $ be a $(d+1)$-dimensional column vector over $ \mathbb{F}_2 $ and $ A $ be a matrix over $ \mathbb{F}_2 $. The indicator function $ \chi_{AX} $ takes value 1 when $ AX = \mathbf{0} $ and 0 otherwise, which indicates an affine space.

\begin{definition}\label{defa}
We denote by $\mathscr{A}$ the set of signatures which have the form $\lambda\cdot\chi_{AX}\cdot\mathfrak{i}^{L_1(X)+L_2(X)+\cdots+L_n(X)}$, where $\mathfrak{i}$ is the imaginary unit, $\lambda\in\mathbb{C}$, $n\in\mathbb{Z}_+$. Each $L_j$ is a 0-1 indicator function of the form $\langle \alpha_j,X \rangle$, where $\alpha_j$ is a $(d+1)$-dimensional vector over $\mathbb{F}_2$, and the dot product $\langle \cdot,\cdot \rangle$ is computed over $\mathbb{F}_2$.
\end{definition}

\begin{definition}\label{defp}
 We denote by $\mathscr{P}$ the set of all signatures which can be expressed as a product of unary signatures, binary equality signatures ($=_2$) and binary disequality signatures ($\neq_2$) on not necessarily disjoint subsets of variables. 
\end{definition}
The following property will be used in our proof, which can be verified from Definition \ref{defa}
 and \ref{defp} directly.
 \begin{lemma}\label{lem:dagonggaocheng}
     Suppose $f,g$ are binary signatures and $M_f=M_g^{-1}$. Then $f\in \mathscr{A}$ (or $\mathscr{P}$) if and only if $g\in \mathscr{A}$ (or $\mathscr{P}$ respectively).
 \end{lemma}
 
 Now we can state the vital theorem that classifies complex-valued Boolean $\ccsp$ problems as follows:

\begin{theorem}[{\cite[Theorem 3.1]{cai2014complexity}}]\label{thm:CSPdichotomy}
Suppose $\mathcal{F}$ is a finite set of signatures mapping Boolean inputs to complex numbers. If $\mathcal{F}\subseteq\mathscr{A}$ or $\mathcal{F}\subseteq\mathscr{P}$, then $\#\csp(\mathcal{F})$ is computable in polynomial time. Otherwise, $\#\csp(\mathcal{F})$ is \#P-hard.
\end{theorem}

Next we introduce the complexity classification for $\ccsp_d(\mathcal{F})$. In this article, we focus on the case that $d=2$ and a special case that $\neq_2\in\mathcal{F}$. Let $\rho_d=\text{e}^{\frac{i\pi}{2d}}$ be a 4d-th primitive root of unity, $T_d=\begin{pmatrix} 
   1 & 0 \\ 
   0 & \rho_d 
   \end{pmatrix},$ and $\mathscr{A}_d^r=\{T_d^rf\mid f\in\mathscr{A}\}$, where $r\in[d]$. 
   A signature $f:\{0,1\}^n\rightarrow\mathbb{C}$ is called local affine if it satisfies $(\bigotimes_{j=1}^nT_2^{\alpha_j})f\in\mathscr{A}$ for any $\alpha\in \su(f)$. The set of all local affine signatures is denoted by $\mathscr{L}$. The dichotomies are stated as follows.

\begin{theorem}[{\cite{cai2018realholantc}}]\label{thm:csp2 dichotomy}
    Suppose $\mathcal{F}$ is a finite set of signatures. If $\mathcal{F}\subseteq\mathscr{A},\mathscr{P},\mathscr{A}_2^1$ or $\mathscr{L}$, then $\ccsp_2(\mathcal{F})$ is polynomial-time computable, otherwise it is \#P-hard.
\end{theorem}

\begin{theorem}[{\cite[Theorem 5.3]{cai2020holantoddarity}}]\label{thm:cspd neq_2 dichotomy}
Let $\mathcal{F}$ be a set of complex-valued signatures. If $\mathcal{F}\subseteq\mathscr{P}$ or $\mathcal{F}\subseteq\mathscr{A}_d^r$ for some $r\in[d]$, the $\ccsp_d(\neq_2,\mathcal{F})$ is tractable; otherwise, $\ccsp_d(\neq_2,\mathcal{F})$ is \#P-hard.
\end{theorem}

\subsubsection{$\hol^c$}
Several important notations and definitions are introduced. 
\begin{definition}
We use following notations.
    \begin{itemize}

    \item $\mathcal{T}$ is the set of all unary and binary signatures.
    
    \item $\mathcal{E}:=\{f\mid\exists \alpha\in\{0,1\}^{arity(f)}\text{ such that }f(x)=0\text{ if }x\notin\{\alpha,\overline{\alpha}\}\}$.

    \item $\mathcal{M}:=\{f\mid   f(x)=0\text{ if }\#_1(x)>1\}$ is the set of weighted matching signatures .

    \item 
$\mathcal{B}:=\{M\mid M^\T\{=_2,\pin_0,\pin_1\}\subseteq\mathscr{A}\}$.
\end{itemize}
\end{definition}

\begin{definition}
    We say a signature set $\mathcal{F}$ is $\mathscr C$-transformable if there exists a $M\in\textbf{GL}_2(\mathbb{C})$ such that $=_2M\in\mathscr{C}$ and $M^{-1}\mathcal{F}\subseteq\mathscr{C}$.
\end{definition}


\begin{theorem}[{\cite[Theorem 59]{backens2021full}}]\label{thm:holantc dichotomy}
Let $\mathcal{F}$ be finite. Then $\hol^c(\mathcal{F})$ is \#P-hard unless:
\begin{enumerate}
\item $\mathcal{F}\subseteq\left \langle \mathcal{T} \right \rangle$;

\item $\mathcal{F}\cup\{\pin_0,\pin_1\}$ is $\mathscr{P}$-transformable. Equivalently, there exists $O\in\mathscr{O}$ such that $\mathcal{F}\subseteq\left \langle O\mathcal{E} \right \rangle$, or $\mathcal{F}\subseteq\left \langle K\mathcal{E} \right \rangle=\left \langle KX\mathcal{E} \right \rangle$;

\item $\mathcal{F}\subseteq\left \langle K\mathcal{M} \right \rangle$ or $\mathcal{F}\subseteq\left \langle KX\mathcal{M} \right \rangle$;

\item $\mathcal{F}\cup\{\pin_0,\pin_1\}$ is $\mathscr{A}$-transformable. Equivalently, there exists $B\in\mathcal{B}$ such that $\mathcal{F}\subseteq B \mathscr{A}$;

\item $\mathcal{F}\subseteq\mathscr{L}$.
\end{enumerate}
In all of the exceptional cases, $\hol^c(\mathcal{F})$ is polynomial-time computable.
\end{theorem}

We remark that, if $\mathcal{F}\subseteq\langle\mathcal{T}\rangle$, then $\hol(\mathcal{F})$ is polynomial-time computable directly by sequential matrix multiplication.

\subsubsection{\ceo}
Let $ f $ be an $\eo$ signature of arity $ 2d $ with $ \text{Var}(f) = \{x_1, x_2, \ldots, x_{2d}\} $. For an arbitrary perfect pairing $ P $ of $ \text{Var}(f) $, say $ P = \{\{x_{i_1}, x_{i_2}\}, \{x_{i_3}, x_{i_4}\}, \ldots, \{x_{i_{2d-1}}, x_{i_{2d}}\}\} $, we define $ \eom[P] $ as the subset of $ \{0,1\}^{2d} $ satisfying that $\eom[P] = \{\alpha \in \{0,1\}^{2d} \mid \alpha_{i_1} \neq \alpha_{i_2}, \ldots, \alpha_{i_{2d-1}} \neq \alpha_{i_{2d}}\}.
$

\begin{definition}\label{def:eoaeop}
  Suppose $f$ is an arity $2d$ $\eo$ signature and $S\subseteq \eoe$. $f|_{S}$ is the restriction of $f$ to $S$, which means when $\alpha\in S$, $f|_{S}(\alpha)=f(\alpha)$, otherwise $f|_{S}(\alpha)=0$.
  
  If for any perfect pairing $P$ of Var$(f)$, $f|_{\eom[P]}\in\mathscr{A}$, then we say that $f$ is $\eom[\mathscr{A}]$.
  
  Similarly, if for any perfect matching $P$ of Var$(f)$, $f|_{\eom[P]} \in \mathscr{P}$, then $f$ is $\eom[\mathscr{P}]$.
\end{definition}

\begin{definition}\label{def:reba}
An $\eo$ signature $f$ of arity $2d$ is called \ba[0](\ba[1] respectively), when the following  recursive conditions are met.
\begin{itemize}
    \item $d=0$: No restriction.
    \item $d\ge 1$: For any variable $x$ in $X=\text{Var}(f)$, there exists a variable $y=\psi(x)$ different from $x$, such that for any $\alpha\in\{0,1\}^X$, if $\alpha_x=\alpha_y=0$($\alpha_x=\alpha_y=1$ respectively) then $f(\alpha)=0$. Besides, the arity $2d-2$ signature $f^{x=0,y=1}$ is 0-rebalancing($f^{x=1,y=0}$ is \ba[1] respectively).
\end{itemize}
For completeness we view all nontrivial signatures of arity 0, which is a non-zero constant, as \ba[0](\ba[1]) signatures. Moreover, an \eo\ signature set $\mathcal{F}$ is said to be \ba[0](\ba[1] respectively) if each signature in $\mathcal{F}$ is \ba[0](\ba[1] respectively). 
\end{definition}

\begin{definition}
    For an $\eo$ signature $f$, if there exists $\alpha,\beta,\gamma\in \su(f)$ and $\delta=\alpha\oplus\beta\oplus\gamma$, such that $\delta\in\eoe$ and $\delta\notin \su(f)$, then we say $f$ is a $\exists3\nrightarrow$ signature. If $\delta\in \eosg$ (or $\delta\in \eosl$) instead, we say $f$ is a $\exists 3\uparrow$ signature (or a $\exists 3\downarrow$ signature). If $f$ is neither a $\exists3\nrightarrow$ signature nor a $\exists 3\uparrow$ signature (or a $\exists 3\downarrow$ signature), we say it is a $\mitsudown$ signature (or a $\mitsuup$ signature).
    \label{def:thupthdown}
\end{definition}

\begin{theorem}[\cite{meng2024p,meng2025fpnp}]
      Let $\mathcal{F}$ be a set of \eo\ signatures. Then $\ceo(\mathcal{F})$ is \#P-hard, unless all signatures in $\mathcal{F}$ are $\mitsuup$ signatures or all signatures in $\mathcal{F}$ are $\mitsudown$ signatures, and $\mathcal{F}\subseteq \eom[\mathscr{A}]$ or $\mathcal{F}\subseteq \eom[\mathscr{P}]$, in which cases it is in $\text{FP}^{\text{NP}}$.
      
 In particular, if the aforementioned condition holds, then $\ceo(\mathcal{F})$ is polynomial-time computable, unless $\mathcal{F}$ is neither \ba[0] nor \ba[1].
 \label{thm:eo dichotomy}
 \end{theorem}

For a signature $f$, we use $f|_{\eo}$ to denote the signature which satisfies that $f|_{\eo}(\alpha)=f(\alpha)$ if $\alpha\in\eoe$, and $f|_{\eo}(\alpha)=0$ otherwise. This notation is consistent with Definition \ref{def:eoaeop} by regarding the symbol \eo\ as $\eoe$. We use $\mathcal{F}|_{\eo}$ to denote the set $\{f|_{\eo}\mid f\in \mathcal{F}\}$. A direct corollary of Theorem \ref{thm:eo dichotomy} is the dichotomy for $\khol$ problems defined by $\eog$ (or $\eol$) signatures.

\begin{corollary}[\cite{meng2025fpnp}]\label{cor:eol eog dichotomy}
    Suppose $\mathcal{F}$ is a set of $\eog$ (or $\eol$ respectively) signatures. Then $\khol(\mathcal{F})$ is \#P-hard, unless all signatures in $\mathcal{F}|_{\eo}$ are $\mitsuup$ signatures or all signatures in $\mathcal{F}|_{\eo}$ are $\mitsudown$ signatures, and $\mathcal{F}|_{\eo}\subseteq \eom[\mathscr{A}]$ or $\mathcal{F}|_{\eo}\subseteq \eom[\mathscr{P}]$, in which cases it is in $\text{FP}^{\text{NP}}$.
\end{corollary}

At last we introduce the dichotomy for $\khol$ problems defined by single-weighted signatures. Suppose $f$ is a signature of arity $k$. If $f$ takes the value $0$ on all input strings whose Hamming weight is not equal to $d$, where $d$ is an integer satisfying $0\leq d\leq k$, then we say $f$ is a \textit{\sw}\ signature.

For each single-weighted signature $f\in \mathcal{F}$ of arity $k$ which can only take non-zero values at Hamming weight $d,0\le d\le k$, let 
\begin{equation}
f_{\to\eo}=\begin{cases}
f\otimes \Delta_0^{2d-k} & 2d\ge k;\\
f\otimes \Delta_1^{k-2d}, & 2d<k.\notag
\end{cases}
\end{equation}
 Let $\mathcal{F}_{\to\eo}=\{f_{\to\eo}|f\in\mathcal{F}\}$. Then we can state the dichotomy for $\khol$ problems defined by single-weighted signatures.
\begin{theorem}[\cite{meng2025fpnp}]\label{thm:single weighted dichotomy}
    Suppose $\mathcal{F}$ is a set of \sw\ signatures. Then $\khol(\mathcal{F})$ is \#P-hard, unless one of the following holds, in which cases it is in $\text{FP}^{\text{NP}}$.
    \begin{enumerate}
        \item All signatures in $\mathcal{F}$ are $\eog$ (or $\eol$ respectively) signatures.  In addition, all signatures in $\mathcal{F}|_{\eo}$ are $\mitsuup$ signatures or all signatures in $\mathcal{F}|_{\eo}$ are $\mitsudown$ signatures, and $\mathcal{F}|_{\eo}\subseteq \eom[\mathscr{A}]$ or $\mathcal{F}|_{\eo}\subseteq \eom[\mathscr{P}]$;
        \item There exist a signature that is not $\eog$ and a signature that is not $\eol$ belonging to $\mathcal{F}$.  In addition, All signatures in $\mathcal{F}_{\to\eo}$ are $\mitsuup$ signatures or all signatures in $\mathcal{F}_{\to\eo}$ are $\mitsudown$ signatures, and $\mathcal{F}_{\to\eo}\subseteq \eom[\mathscr{A}]$ or $\mathcal{F}_{\to\eo}\subseteq \eom[\mathscr{P}]$.
    \end{enumerate}
\end{theorem}

There is another interesting class of signatures. Suppose $\mathcal{F}$ is a set of signatures. It is called vanishing if the partition functions of all instances of $\hol(\mathcal{F})$ are 0. We characterize this class in Lemma \ref{lem:vanish is eosg or eosl}.

\section{Main results}\label{main}
In this section, we present our main results in detail. The first result is the \textit{generalized decomposition lemma}, which can be applied within the complex-valued $\hol$ framework and provides a powerful tool to make complexity classifications. We give the proof of the following theorem in Section \ref{sec:decomposition lemma}.

\begin{theorem}[Decomposition lemma]\label{thm:decomposition lemma}
    Let $\mathcal{F}$ be a set of signatures and $f$, $g$ be two signatures. Then

$$\hol(\mathcal{F},f,g)\equiv_T\hol(\mathcal{F},f\otimes g)$$ 
holds unless $\widehat{\mathcal{F}}'=\widehat{\mathcal{F}}\cup\{\widehat{f}\otimes\widehat{g}\}$ only contains $\eog$ signatures (or $\eol$ signatures respectively). 
In the latter situation, if all signatures in $\widehat{\mathcal{F}}'|_{\eo}$ are $\mitsuup$ signatures or all signatures in $\widehat{\mathcal{F}}'|_{\eo}$ are $\mitsudown$ signatures, and $\widehat{\mathcal{F}}'|_{\eo}\subseteq \eom[\mathscr{A}]$ or $\widehat{\mathcal{F}}'|_{\eo}\subseteq \eom[\mathscr{P}]$, then $\hol(\mathcal{F},f\otimes g)$ is in $\pnp$. Otherwise it is \#P-hard.
\end{theorem}

The second result is the dichotomy for complex-valued $\holodd$. The proof of it is in Section \ref{sec:proof of main}. 

\begin{theorem}[Dichotomy for $\holodd$]\label{thm:main theorem}
    Let $\mathcal{F}$ be a set of signatures that contains a non-trivial signature of odd arity, then $\hol(\mathcal{F})$ is \#P-hard unless:
    \begin{enumerate}
        \item $\widehat{\mathcal{F}}$ only contains $\eog$ signatures (or $\eol$ signatures respectively). All signatures in $\widehat{\mathcal{F}}|_{\eo}$ are $\mitsuup$ signatures or all signatures in $\widehat{\mathcal{F}}|_{\eo}$ are $\mitsudown$ signatures, and $\widehat{\mathcal{F}}|_{\eo}\subseteq \eom[\mathscr{A}]$ or $\widehat{\mathcal{F}}|_{\eo}\subseteq \eom[\mathscr{P}]$;
        \item All signatures in $\widehat{\mathcal{F}}$ are single-weighted. There exist a signature that is not $\eog$ and a signature that is not $\eol$ belonging to $\widehat{\mathcal{F}}$.  In addition, all signatures in $\widehat{\mathcal{F}}_{\to\eo}$ are $\mitsuup$ signatures or all signatures in $\widehat{\mathcal{F}}_{\to\eo}$ are $\mitsudown$ signatures, and $\widehat{\mathcal{F}}_{\to\eo}\subseteq \eom[\mathscr{A}]$ or $\widehat{\mathcal{F}}_{\to\eo}\subseteq \eom[\mathscr{P}]$;

        \item $\mathcal{F}\subseteq\langle\mathcal{T}\rangle$;

        \item $\widehat{\mathcal{F}}\subseteq\langle \mathcal{M}\rangle$ or $\widehat{\mathcal{F}}\subseteq\langle X\mathcal{M}\rangle$;

        \item $\mathcal{F}$ is $\mathscr{A}$-transformable;

        \item $\mathcal{F}$ is $\mathscr{P}$-transformable;

        \item $\mathcal{F}$ is $\mathscr{L}$-transformable;
    \end{enumerate}
    In case 1, 2, $\hol(\mathcal{F})$ is in $\pnp$; in case 3-7, $\hol(\mathcal{F})$ is polynomial time computable.  We denote Case 1-7 as condition ($\mathcal{PC}$).
\end{theorem}

\begin{remark}

    In case 1, 2 of condition ($\mathcal{PC}$), there are parts of situations where the problem is actually polynomial time computable. One typical example is the vanishing signature. By Lemma \ref{lem:vanish is eosg or eosl}, if $\widehat{\mathcal F}$ is $\eosl$ or $\eosg$, or equivalently $\mathcal{F}$ is vanishing, then $\hol(\mathcal{F})$ is polynomial-time computable. We do not address these classes in this article, and the detailed information can be found in \cite{meng2024p,meng2025fpnp}.
\end{remark}


The proof of Theorem \ref{thm:main theorem} consists of two parts. Firstly, we prove a dichotomy for complex-valued \hol\ when $\pin_0$ is available, stated as Lemma \ref{lem:hol0 dichotomy} in the following. In particular, we remark that it is a traditional FP vs. \#P dichotomy. The reduction map of it is presented as Figure \ref{fig:hol d0 map}. Then we prove the dichotomy for complex-valued $\holodd$ based on this dichotomy, whose reduction map is presented as Figure \ref{fig:hol odd map}.

\begin{lemma}
\label{lem:hol0 dichotomy}
    Let $\mathcal{F}$ be a set of signatures, then $\hol(\mathcal{F},\Delta_0)$ is \#P-hard unless:
    \begin{enumerate}
        \item $\mathcal{F}\subseteq\langle\mathcal{T}\rangle$;

        \item $\mathcal{F}\subseteq\langle K\mathcal{M}\rangle$ or $\mathcal{F}\subseteq\langle KX\mathcal{M}\rangle$;

        \item $\mathcal{F}\cup\{\Delta_0\}$ is $\mathscr{A}$-transformable;

        \item $\mathcal{F}$ is $\mathscr{P}$-transformable;

        \item $\mathcal{F}\cup\{\pin_0\}$ is $\mathscr{L}$-transformable;
    \end{enumerate}
    in which cases, $\hol(\mathcal{F},\Delta_0)$ is polynomial-time computable.
\end{lemma}

\begin{figure}
    \begin{center}
\begin{tikzpicture}[
  minimum size=5mm,
  sibling distance=4cm, edge from parent/.style={draw,-latex}
  ]
\node[draw] (d0){\hol\ with $\pin_0$}
child {node[draw] (=4)[xshift=-0.2cm]{\hol\ with $=_4$}
child {node[draw](csp2)[yshift=-3cm]{$\ccsp_2$}}}
child {node[draw,align=left](d03){\hol\ with $\pin_0$ and\\ irreducible ternary $f$}
child {node[draw,align=left](ghz){\hol\ with $\pin_0$\\and $f$ of GHZ type}
child {node[draw,align=left](symghz){\hol\ with \\symmetric $f$\\ of GHZ type}
child {node[draw](csp){\ccsp}}
}}
child{node[draw,align=left](w){\hol\ with $\pin_0$\\and $f$ of W type}
child{node[draw,align=left](-1100)[xshift=2cm]{\khol\ with $[1,1]$\\ and $[-1,1,0,0]$\\(or $[0,0,1,-1]$)}}
}
}
child {node[draw](holc) {$\hol^c$}}
;
\draw[-latex] (w) ->  (symghz);
\draw[-latex] (ghz) ->  (w);
\draw[-latex] (w) ->  (holc);
\draw[-latex] (-1100) ->  (holc);
\draw[-latex] (-1100) ->  (symghz);
\end{tikzpicture}
\end{center}
    \caption{The reduction map for \hol\ with $\pin_0$. We use the term 'A$\to$B' to mean the corresponding reduction in our proof. It is noteworthy that a node may have multiple outgoing arcs, and in this case at least one of the reduction holds.}
    \label{fig:hol d0 map}
\end{figure}
\begin{figure}
    \begin{center}
\begin{tikzpicture}[
  minimum size=5mm,
  sibling distance=4cm, edge from parent/.style={draw,-latex}
  ]
\node[draw,align=left](odd){$\holodd$}
child {node[draw] (d0){\hol\ with $\pin_0$}}
child {node[draw,align=left](kd0) {\khol\ with\\ $\pin_0$(or $\pin_1$)}
child {node[draw,align=left](kfc){\khol\ with\\ $\pin_0$ and $\pin_1$}  
child{node[draw](d0d1){$\hol^c$ with $\neq_2$}}
}
child{node[draw](eo)[xshift=-1cm,yshift=-1.5cm]{\ceo}}
}
child {node[draw,align=left](=k)[yshift=-1.5cm] {\khol\ with \\$[a,0,\dots,0,b]_{k},$\\$ab\neq0,k\ge 3$}
child {node[draw](cspk){$\ccsp_k$ with $\neq_2$}}}
;
\draw[-latex] (kd0) ->  (d0);
\draw[-latex] (kfc) ->  (eo);
\draw[-latex] (kfc) ->  (d0);
\draw[-latex] (kfc) ->  (=k);
\end{tikzpicture}
\end{center}
    \caption{The reduction map for $\holodd$. We use the term 'A$\to$B' to mean the corresponding reduction in our proof. It is noteworthy that a node may have multiple outgoing arcs, and in this case at least one of the reduction holds.}
    \label{fig:hol odd map}
\end{figure}

Now we present the proof strategy for Theorem \ref{thm:main theorem}. The proof combines the techniques from the dichotomy for real-valued $\holodd$ \cite{cai2020holantoddarity} and that for complex-valued $\hol^c$ \cite{backens2021full}. The structure of the proof is also a combination of those from the aforementioned references, 
and basically consists of four parts.  In general, the structure can be stated as follows.

\begin{enumerate}
    \item Realize a non-trivial unary signature, and transform it into $\pin_0$ through holographic transformation.
    \item Use $\pin_0$ and self-loops to obtain an irreducible ternary signature $f$.
    \item Realize a symmetric signature $h$ of GHZ type with $f$.
    \item Transform $h$ into $[1,0,0,1]$ through holographic transformation, and reduce a certain \ccsp\ to this problem.
\end{enumerate}

Step 1 comes from \cite{cai2020holantoddarity}, stated as Lemma \ref{lem:1fenlei}. In Step 1, the proof uses self-loop to reduce the arity of the odd signature by 2 at a time, while keeping the obtained signature non-trivial. If this process fails, we deal with this case by \cite[Theorem 5.3]{cai2020holantoddarity} with some slight modifications. Otherwise we may obtain a unary signature.

Step 2 also comes from \cite{cai2020holantoddarity}. In Step 2, we use $\pin_0$, self-loops and the technique of Unique Prime Factorization (UPF) to reduce the arity of an irreducible signature as in \cite{cai2020holantoddarity}. Again we can realize an irreducible ternary signature $f$, or deal with other cases by the dichotomy in \cite{backens2021full} or \cite[Lemma 5.2]{cai2018realholantc} with some slight modifications on the origin proof. This part corresponds to the first part  of Section \ref{sec:hol0}. 

Step 3 comes from \cite{backens2021full}. We do not further reduce the arity of $f$ to 2 as this part of proof in \cite{cai2020holantoddarity} is quite dependent on the properties of real numbers. Rather, we use the knowledge from quantum entanglement theory to classify such $f$ into different types.  We also use the methods developed in \cite{backens2021full} to realize a symmetric signature $h$ of GHZ type for different types of signatures. This part corresponds to the last part  of Section \ref{sec:hol0}. 

Step 4 involves the method in \cite[Theorem 5]{CHL2012symholantc}, which is also employed in \cite{backens2021full}. The only difference between this step and the corresponding proof in \cite{backens2021full} is the inability to realize $\pin_1$. However, this does not introduce any challenges. This part corresponds to the middle part  of Section \ref{sec:hol0}. 

On the other hand, we remark that nontrivial obstacles occur in Step 1, 2 and 3. 

In Step 1, if the obtained signature is a multiple of $[1,\ii]$ or $[1,-\ii]$, then the obtained unary signature may not be transformed into $\pin_0$. This special case does not appear in the setting of real-valued signatures, and consequently we need to analyze it separately. The analysis of this part includes a number of reductions, which induces a complicated  reduction map, as presented in Figure \ref{fig:hol odd map}.  This obstacle is overcome in Section \ref{sec:khol}. We also remark that in these reductions, the dichotomies respectively for $\eog$, $\eol$ and \sw\ signatures are employed, hence the complexity class $\pnp$ is also introduced.

In Step 2, the UPF method employs the decomposition lemma from \cite{LinWang2018plushol}. As the original lemma does not hold for all complex-valued signature sets, this motivated us to prove the generalized version of decomposition lemma, presented as Theorem \ref{thm:easydecompo}.

In Step 3, due to the absence of $\pin_1$, we are not able to realize a symmetric $h$ of GHZ type for some $f$ of W type. Furthermore, for some specific $f$, we may not even simulate $\pin_1$ through polynomial interpolation. By carefully analyzing each possible case, we prove that for such $f$, $\widehat{f}$ is restricted to be either $[-1,1,0,0]$ or $[0,0,1,-1]$. For this specific case, we prove that the form of each other signature is also restricted to a specific form, otherwise again symmetric $h$ of GHZ type or $\pin_1$ can be realized. This proof is challenging, as it involves detailed classification for signatures, a special holographic transformation to restrict the form of each signature, the UPF method for reducing the arity and a complicated analysis for the support of each signature. Finally, the only case left is when all signatures in $\mathcal{F}$ take this specific form, which induces the tractability. The proof of this part is also presented in Figure \ref{fig:hol d0 map}, and involves Lemma \ref{lem:no up triangle}-\ref{lem:yi dingyou unary or binary or tenary}.

\section{Decomposition lemma}\label{sec:decomposition lemma}

In this chapter, we present the \textit{generalized decomposition lemma}, which demonstrates that, under certain conditions, $\hol(f,g,\mathcal{F})\equiv_T\hol(f\otimes g,\mathcal{F})$, otherwise the computational complexity of $\hol(f\otimes g,\mathcal{F})$ is known. The proof of this result relies on some preliminary lemmas, which are presented in the following.

\begin{lemma}[{\cite[Lemma 3.1]{LinWang2018plushol}}]
    For any signature set $\mathcal{F}$ and signature $f$,
    
$$\hol(\mathcal{F},f)\leq_T\hol(\mathcal{F},f^{\otimes d})$$
    
    for all $d\geq1$.
    \label{lem:singledecomposition}
\end{lemma}

\begin{lemma}[{\cite[Lemma 3.2]{LinWang2018plushol}}]\label{lem:nonzero instance decompose}
Let $\mathcal{F}$ be a set of signatures, and $f$, $g$ be two signatures. Suppose that
there exists an instance $I$ of $\hol(\mathcal{F},f, g)$ such that Z$(I)\neq0$, and the number
of occurrences of $g$ in $I$ is greater than that of $f$. Then

$$\hol(\mathcal{F},f,f\otimes g)\leq_T\hol(\mathcal{F},f\otimes g).$$
\end{lemma}

\begin{lemma}[{\cite[Corollary 3.3]{LinWang2018plushol}}]\label{lem: not vanishing decompose}
    Let $\mathcal{F}$ be a set of signatures, and $f$, $g$ be two signatures. If $g$ is
not vanishing, then

$$\hol(\mathcal{F},f,f\otimes g)\leq_T\hol(\mathcal{F},f\otimes g).$$
\end{lemma}

\begin{remark}\label{remark:holo doesn't change decomposition}
    As mentioned in \cite{cai2020beyond}, the unique prime factorization of a signature $f$ naturally divides $\var(f)$ into several sets, each corresponding to a prime factor of $f$. It can be easily verified that the partition of $\var(f)$ remains unchanged under holographic transformation (even SLOCC). Therefore, we can apply the three lemmas above in the setting of $\khol$.  
\end{remark}

\begin{lemma}\label{lem:self-loop to pure1 string}
Suppose $f$ is a signature and $\alpha\in\su(f)$. Suppose $\#_0(\alpha)-\#_1(\alpha)=k>0$ (or $\#_1(\alpha)-\#_0(\alpha)=k>0$), then a signature $g$ of arity $k$  can be obtained by adding self-loops using $\neq_2$ on $f$ such that $g(\textbf{0})\neq0$ (or $g(\textbf{1})\neq0$ respectively).
\end{lemma}
\begin{proof}
    We only prove the case where $\#_0(\alpha)-\#_1(\alpha)=k>0$. The other case is similar. 
    
    Let $\ari(f)=d$. If $d=1$ or $2$, the statement is already true. Suppose $d>2$. If $\#_1(\alpha)=0$, the statement also holds. If $\#_1(\alpha)>0$, without loss of generality, assume the first three bits of $\alpha$ are 0, 1 and 0.

    We write $\alpha$ as $\alpha=010\alpha'$. Consider $\beta=100\alpha'$ and $\gamma=001\alpha'$. Let $f(\alpha)=x$, $f(\beta)=y$ and $f(\gamma)=z$. If $x+y\neq0$, we add a self-loop by $\neq_2$ on the first two variables of $f$ and obtain a signature $g$ such that $g(0\alpha')\neq0$. We have $\#_0(0\alpha')-\#_1(0\alpha')=k$. If $x+z\neq0$ (or $y+z\neq0$), we similarly add a self-loop on the second and the third (or the first and the third) variables of $f$ to obtain a $g$. If $x+z=x+y=y+z=0$, then $x=0$, which contradicts $\alpha\in\su(f)$. Therefore, we obtain a arity $d-2$ signature $g$ satisfying the condition that there exists a string $\alpha\in\su(g)$ such that $\#_0(\alpha)-\#_1(\alpha)=k>0$.

    Noticing that the process above decreases $\#_0(\alpha)$ and $\#_1(\alpha)$ by 1. Repeating this process, we can decrease $\#_1(\alpha)$ to 0. Therefore, the lemma is proved.
\end{proof}



\begin{corollary}\label{cor:hw> construct Delta1}
     Suppose $f$ is an $\eosg$ (or $\eosl$) signature, by adding self-loop by $\neq_2$ we can always obtain $\lambda\Delta_1^{\otimes r}$ (or $\lambda\Delta_0^{\otimes r}$), where $r>0$ is an integer. 
\end{corollary}
\begin{proof}
    We only prove the case where $f$ is an $\eosg$ signature. The other case is similar.

    If $f$ is a unary signature, then $f$ is $\lambda\Delta_1$, where $\lambda\neq0$ is a constant, and the statement holds. If $f$ is a binary signature, then $f$ is $\lambda\Delta_1^{\otimes2}$, where $\lambda\neq0$ is a constant. By Lemma \ref{lem:singledecomposition} we are done. In the following, we assume $\ari(f)=d>2$.

    Assume that $\alpha$ is the string in $\su(f)$ such that $\#_1(\alpha)-\#_0(\alpha)$ takes the minimum value. By Lemma \ref{lem:self-loop to pure1 string} we can obtain a signature $g$ such that $g(1_{\#_1(\alpha)-\#_0(\alpha)})\neq0$. By the constructing process in Lemma \ref{lem:self-loop to pure1 string} and the fact that $\#_1(\alpha)-\#_0(\alpha)$ is minimal, we have $g=\lambda\Delta_1^{\otimes(\#_1(\alpha)-\#_0(\alpha))}$.
\end{proof}
Symmetric vanishing signatures have been characterized in \cite{cai2013vanishing}. The following lemma characterizes vanishing signatures not necessarily to be symmetric.
\begin{lemma}\label{lem:vanish is eosg or eosl}
    Suppose $f$ is a signature and $\widehat{f}=K^{-1}f$. Then the following two statements are equivalent:

    1. $f$ is vanishing;

    2. $\widehat{f}$ is a $\eosg$ or $\eosl$ signature.
\end{lemma}
\begin{proof}
    If $\widehat{f}$ is an $\eosg$ or $\eosl$ signature, the partition function of any instance of $\khol(\widehat{f})$ is 0, that is, $f$ is vanishing.

    In the other direction, we prove the contrapositive of the proposition. Suppose $\widehat{f}$ is not an $\eosg$ or $\eosl$ signature. 
    
    The first case is that there is a string $\alpha\in\su(\widehat{f})$ such that $\alpha\in\eoe$, then $\widehat{f}$ is of even arity. If $\ari(\widehat{f})=d\geq4$, let $\alpha=010\alpha'$, $\beta=100\alpha'$ and $\gamma=001\alpha'$. By the same approach used in the proof of Lemma \ref{lem:self-loop to pure1 string}, we can obtain a signature $g$ of arity $d-2$ satisfying that $\su(g)\cap\eoe\neq\emptyset$. Repeating this process (including 0 times), we can obtain a binary signature $h$, which is not an $\eosg$ or $\eosl$ signature. Suppose $M_h=\begin{pmatrix}
        a & b\\
        c & d
    \end{pmatrix}$, where $b\neq0$. If $b+c\neq0$, then adding a self-loop by $\neq_2$ yields a non-zero value, which implies that $f$ is not vanishing. Assume $b+c=0$. Connecting edges of $h$ to another copy of $h$ using $\neq_2$, the values of two possible instances are the trace of $M_hXM_hX$ and $M_hXM_h^\T X$, which are $2ad+2b^2$ and $2ad-2b^2$. Since $b\neq0$, at least one of the two values is non-zero. Therefore, $f$ is not vanishing.

    The second case is that $\su(f)\cap\eoe=\emptyset$, then both $\su(f)\cap\eosg$ and $\su(f)\cap\eosl$ are non-empty. Suppose $\alpha\in\su(f)\cap\eosg$ with $\#_1(\alpha)-\#_0(\alpha)=k>0$, and $\beta\in\su(f)\cap\eosl$ with $\#_0(\beta)-\#_1(\beta)=l>0$. Then $\su(f^{\otimes(k+l)})\cap\eoe\neq\emptyset$, since the string consisting of $l$ copies of $\alpha$ and $k$ copies of $\beta$ is in its support. Therefore, the second case is reduced to the first case.

    In summary, the lemma is proved.
\end{proof}

We now prove Theorem \ref{thm:decomposition lemma}. 
\begin{proof}[Proof of Theorem \ref{thm:decomposition lemma}]
 To prove $\hol(\mathcal{F},f,g)\leq_T\hol(\mathcal{F},f\otimes g)$, it is sufficient to prove $\khol(\widehat{\mathcal{F}},\widehat{f},\widehat{g})\leq_T\khol(\widehat{\mathcal{F}},\widehat{f}\otimes\widehat{g})$ by Remark \ref{remark:holo doesn't change decomposition}. Now we consider the equivalent problem $\khol(\widehat{\mathcal{F}},\widehat{f}\otimes\widehat{g})$. 

     If $\khol(\widehat{\mathcal{F}}\cup\{\widehat{f}\otimes\widehat{g}\})$ only contains $\eog$ signatures (or $\eol$ signatures respectively), then by Corollary \ref{cor:eol eog dichotomy}, the complexity of $\hol(f\otimes g,\mathcal{F})$ is known, which is either in $\pnp$ or \#P-hard and the classification criterion is explicit.
     
     Now we assume that there exist signatures $\widehat{h},\widehat{h'}\in \widehat{\mathcal{F}}\cup\{\widehat{f}\otimes\widehat{g}\}$ such that  $\su(\widehat{h})\not\subseteq\eog$ and  $\su(\widehat{h'})\not\subseteq\eol$. 
      According to Lemma \ref{lem:self-loop to pure1 string}, by adding self-loops on $\widehat{h}$ and $\widehat{h'}$ we can obtain two signatures $h_0$ of arity $d_0$ and $h_1$ of arity $d_1$, such that $h_0(\textbf{0})\neq0$ and $h_1(\textbf{1})\neq0$.

If $f$ is not vanishing, by Lemma \ref{lem: not vanishing decompose} we have $\hol(g,f\otimes g,\mathcal{F})\leq_T\hol(f\otimes g,\mathcal{F})$. 
        
        If $f$ is vanishing, by Lemma \ref{lem:vanish is eosg or eosl} we have $\widehat{f}$ is an $\eosl$ (or $\eosg$) signature. By Corollary \ref{cor:hw> construct Delta1}, we can obtain $\lambda\Delta_0^{\otimes r}$ (or $\lambda\Delta_1^{\otimes r}$), where $\lambda,r\neq 0$, by adding self-loops on $\widehat{f}$. By connecting $d_1$ copies of $\lambda\Delta_0^{\otimes r}$ to $r$ copies of $h_1$ (or  $d_0$ copies of $\lambda\Delta_1^{\otimes r}$ to $r$ copies of  $h_0$) via $\neq_2$, we construct an instance $I$ with $Z(I)=\lambda^{d_1}h_1^r(\textbf{1})$ (or $Z(I)=\lambda^{d_0}h_0^r(\textbf{0})$). In other words, we actually construct an instance consisting of a positive number copies of $f$ and some signatures from $ \mathcal{F}\cup\{f\otimes g\}$ that yields a non-zero partition function. By Lemma \ref{lem:nonzero instance decompose}, we have $\hol(g,f\otimes g,\mathcal{F})\leq_T\hol(f\otimes g,\mathcal{F})$.

        The analysis for $g$ is similar. Therefore, $\hol(f,g,\mathcal{F})\leq_T\hol(f\otimes g,\mathcal{F})$.
\end{proof}

 \section{Dichotomy for $\holodd$}\label{oddsi}

 In this chapter, we prove the dichotomy for complex-valued $\holodd$. We emphasize that in the proofs of this chapter, we often normalize signatures by default for convenience, since this operation does not affect the computational complexity. For example, we write $\lambda\pin_0,\lambda\neq 0$ as $\pin_0$ by default. 
 By Theorem \ref{thm:decomposition lemma}, we can always assume the signature in $\mathcal{F}$ is irreducible in this section. We commence with the following lemmas, which separate $\holodd$ into several cases.

\begin{lemma}
    Suppose $f\in\mathcal{F}$ is a non-trivial signature of odd arity. Then one of the following statements holds:
    \begin{enumerate}
        \item There exists some matrix $Q\in\mathbb{C}^{2\times2}$ such that $\hol(Q\mathcal{F},\pin_0)\le_T \hol(\mathcal{F})$;
        \item $\khol(\widehat{\mathcal{F}},\pin_0)\le_T \hol(\mathcal{F})$;
        \item $\khol(\widehat{\mathcal{F}},\pin_1)\le_T \hol(\mathcal{F})$;
        \item $\khol(\widehat{\mathcal{F}},[a,0,\dots,0,b]_{2k+1})\le_T \hol(\mathcal{F})$ for some $k\in\mathbb{N}_+$, where $ab\neq 0$.
    \end{enumerate}
    \label{lem:1fenlei}
\end{lemma}
\begin{proof}
    We prove this lemma by induction. Suppose $f$ is of arity $2k-1,k\in \mathbb{N}_+$. When $k=1$, $f$ is unary. If $f$ is a multiple of $(1,\ii)$ or $(1,-\ii)$, then $\widehat{f}=\pin_0$ or $\widehat{f}=\pin_1$ respectively, which leads to Case 2 or 3. Now suppose $f=(a,b)$ is not a multiple of $(1,\ii)$ or $(1,-\ii)$. Let $Q=\frac{1}{\sqrt{a^2+b^2}}\begin{pmatrix}
        a & b \\
        b & -a\end{pmatrix}$. It can be verified that $Q^\T Q=I$ and $Q(a,b)^\T=(1,0)=\pin_0$. Consequently we have $\hol(Q\mathcal{F},\pin_0)\le_T \hol(\mathcal{F})$.

        Now suppose one of the statements always holds for all $k<n+1,n\in \mathbb{N}_+$. If $k=n+1$, then $f$ is of arity $2n+1$.  If $\widehat{\partial_{ij}}\widehat{f}\not\equiv 0$ for some $i,j$, we are done by the induction hypothesis. Otherwise by Lemma \ref{lem:wangjuqiuwangle}, we have $\widehat{f}=[a,0,\dots,0,b]_{2n+1}$. If $ab\neq 0$, it leads to Case 4. Otherwise by Theorem \ref{thm:decomposition lemma} we can obtain a unary signature and are done by induction.
\end{proof}

In the following, we classify the complexity of each case in Lemma \ref{lem:1fenlei} respectively. In Section \ref{sec:hol0}, we deal with Case 1 and prove Lemma \ref{lem:hol0 dichotomy}. In Section \ref{sec:khol}, we deal with other cases. In Section \ref{sec:proof of main} we present the proof of Theorem \ref{thm:main theorem}.

\subsection{A dichotomy for $\hol$ problems when $\pin_0$ is available}\label{sec:hol0}
In this section, we prove Lemma \ref{lem:hol0 dichotomy}. The reductions appearing in this section are summarized as a map in Figure \ref{fig:hol d0 map}.

If $\mathcal{F}\subseteq\mathcal{T}$, $\hol(\mathcal{F},\pin_0)$ is polynomial-time computable. As a result, we may assume there exists an irreducible signature $f\in \mathcal{F}$ with $\text{arity}(f)\ge3$. By replacing the original decomposition lemma by our generalized version Theorem \ref{thm:decomposition lemma}, the following result, which previously only holds for real-valued \hol\ \cite[Lemma 1.8]{cai2020holantoddarity}, can be extended to complex-valued \hol.

\begin{lemma}
    Suppose $f\in \mathcal{F}$ is irreducible with $\text{arity}(f)\ge3$. Then one of the following holds.
    \begin{enumerate}
        \item There is a ternary irreducible $g\in \gc(f,\pin_0)$;
        \item There is a quaternary $g\in \gc(f,\pin_0)$ satisfying $M_g=\begin{pmatrix}
        a & 0 & 0 & b\\
        0 & 0 & 0 & 0\\
        0 & 0 & 0 & 0\\
        c & 0 & 0 & d\\
    \end{pmatrix}, ad-bc\neq 0$;
        \item $\pin_1\in \gc(f,\pin_0)$.
    \end{enumerate}
    \label{lem:2fenlei}
\end{lemma}

The complexity of the third situation can be classified by Theorem \ref{thm:holantc dichotomy}. The following lemmas  can be used to deal with the second situation. Combining them with Theorem \ref{thm:csp2 dichotomy}, the complexity classification for the second situation is done.

\begin{lemma}[{\cite[Part of Lemma 2.40]{CF2017plCSP}}]\label{lem:hol =4 csp2 part1}
    Suppose $f\in \gc(\mathcal{F})$ satisfying $M_f=\begin{pmatrix}
        a & 0 & 0 & b\\
        0 & 0 & 0 & 0\\
        0 & 0 & 0 & 0\\
        c & 0 & 0 & d\\
    \end{pmatrix}, ad-bc\neq 0$. Then $\hol(\mathcal{F},=_4)\equiv_T\hol(\mathcal{F})$.
\end{lemma}

\begin{lemma}[{\cite[Lemma 5.2]{cai2018realholantc}}]\label{lem:hol =4 csp2 part2}
    Suppose $(=_4)\in\mathcal{F}$. Then $\hol(\mathcal{F})\equiv_T\ccsp_2(\mathcal{F})$.
\end{lemma}

We now focus on the first situation. First, we show that the complexity can be classified if there exists a symmetric signature of GHZ type.
A symmetric signature of GHZ type is denoted as the generic case in \cite{CHL2012symholantc}. We present the following lemmas from \cite{CHL2012symholantc}. It is noteworthy that the concept of $\omega$-normalized signature is important in the following known lemmas, but it does not need to be understood in our proof.

\begin{lemma}[{\cite{CHL2012symholantc}}]
    Suppose $g$ is a symmetric signature of GHZ type, then there exists $M\in\mathbf{GL}_2(\mathbb{C})$ such that $g=M[1,0,0,1]$.
    \label{lem:symGHZisgeneric}
\end{lemma}

\begin{lemma}[{\cite{CHL2012symholantc}}]
 Suppose $[y_0,y_1,y_2]$ is $\omega$-normalized and nondegenerate. If $y_0 = y_2 = 0$,
 further assume that $[a,b]\in \mathcal{G}_1$ is $\omega$-normalized and
 satisfies $ab\neq 0$. Then,
 
    $$\hol(\mathcal{G}_1,[y_0,y_1,y_2] \mid \mathcal{G}_2,=_3)\equiv_T\ccsp(\mathcal{G}_1,\mathcal{G}_2,[y_0,y_1,y_2])$$
    \label{lem:holtocsp}
\end{lemma}

 Let $\Omega_3=\{\omega\mid \omega^3=1\}$. To apply Lemma \ref{lem:holtocsp}, the following facts would be useful.

\begin{lemma}[{\cite{CHL2012symholantc}}]
$[0,1,0]$ is $\omega$-normalized.

     For any symmetric binary signature $f$, there exists $M_{\omega}=\begin{pmatrix}
        1 & 0\\
        0 & \omega
    \end{pmatrix}, \omega\in \Omega_3$ such that $fM_{\omega}$ is $\omega$-normalized.
    
     For any unary signature $f=[a,b], ab\neq 0$, there exists $M_{\omega}=\begin{pmatrix}
        1 & 0\\
        0 & \omega
    \end{pmatrix}, \omega\in \Omega_3$ such that $fM_{\omega}$ is $\omega$-normalized.
    \label{lem:wnormal}
\end{lemma}

Now we are ready to prove the following lemma.

\begin{lemma} \label{lem:symmetric GHZ goto CSP}
Suppose $g$ is a symmetric signature of GHZ type, then there exists an invertible matrix $N$ such that 
$$
\hol(\mathcal{F},g,\pin_0)\equiv_T\ccsp(=_2N,\pin_0N, N^{-1}\mathcal{F}, N^{-1}\pin_0)$$
\label{lem:symGHZto=3}
\end{lemma}
\begin{proof}
    We see $\hol(\mathcal{F},g,\Delta_0)$ equivalently as $\hol(=_2\mid\mathcal{F},g,\Delta_0)$. By connecting $\pin_0$ to $=_2$, we realize $\pin_0$ on the LHS. Since $g$ is a symmetric signature of GHZ type, by Lemma \ref{lem:symGHZisgeneric} there exists an  invertible matrix $M=\begin{pmatrix}
        a & b\\
        c & d
    \end{pmatrix}$ such that $g=M[1,0,0,1]$. Consequently by  Theorem \ref{thm:holographic transformation equivalence} we have
\begin{align*}
 \hol(\mathcal{F},g,\pin_0)&\equiv_T\hol(=_2 \mid \mathcal{F},g,\pin_0)\\
 &\equiv_T\hol(=_2,\pin_0 \mid \mathcal{F},g,\pin_0)\\
 &\equiv_T\hol(=_2M,\pin_0M \mid M^{-1}\mathcal{F},=_3,M^{-1}\pin_0)\\
 &\equiv_T\hol([a^2+c^2,ab+cd,b^2+d^2],[a,b] \mid M^{-1}\mathcal{F},=_3,M^{-1}\pin_0)
\end{align*}
If $a^2+c^2=b^2+d^2=0$ does not hold, then by Lemma \ref{lem:wnormal} there exists $M_{\omega}=\begin{pmatrix}
        1 & 0\\
        0 & \omega
    \end{pmatrix}, \omega\in \Omega_3$ such that $[a^2+c^2,ab+cd,b^2+d^2]M_\omega$ is $\omega$-normalized. Otherwise, $a^2+c^2=b^2+d^2=0$. If $a=0$, then $c=0$ as well, which contradicts that $M$ is invertible. Consequently $a\neq 0$ and similarly $b\neq0$. Again by Lemma \ref{lem:wnormal} there exists $M_{\omega}=\begin{pmatrix}
        1 & 0\\
        0 & \omega
    \end{pmatrix}, \omega\in \Omega_3$ such that $[a,b]M_\omega=[a,b\omega]$ is $\omega$-normalized and $ab\omega\neq0$. Furthermore, $[0,ab+cd,0]M_\omega=\omega(ab+cd)[0,1,0]$ is also $\omega$-normalized. Consequently, in both cases Lemma \ref{lem:holtocsp} can be applied and we have
\begin{align*}
 \hol(\mathcal{F},g,\pin_0)
 &\equiv_T\hol(=_2M,\pin_0M \mid M^{-1}\mathcal{F},=_3,M^{-1}\pin_0)\\
 &\equiv_T\hol(=_2MM_\omega,\pin_0MM_\omega \mid M_\omega^{-1}M^{-1}\mathcal{F},=_3,M_\omega^{-1}M^{-1}\pin_0)\\
 &\equiv_T\ccsp(=_2MM_\omega,\pin_0MM_\omega, M_\omega^{-1}M^{-1}\mathcal{F}, M_\omega^{-1}M^{-1}\pin_0)
\end{align*}
Let $N=MM_\omega$, and the proof is completed.
\end{proof}

Noticing that by Lemma \ref{lem:symmetric GHZ goto CSP}, once a symmetric ternary signature of GHZ type is realized, we can make complexity classification by $\ccsp$ dichotomy. Backens also presents methods for symmetrization \cite{backens2017holant+,backens2021full}. Using those methods, we can realize a symmetric signature of GHZ type in most situations.

\begin{lemma}[{\cite{backens2017holant+}}]
Suppose $f$ is of GHZ type. Then there exists an irreducible symmetric signature $h\in\gc(\{f\})$.
    \label{lem:GHZsym}
\end{lemma}
\begin{lemma}[{\cite{backens2017holant+}}]
Suppose $f$ is of W type and $f\notin \langle K\mathcal{M}\rangle\cup\langle KX\mathcal{M}\rangle$. Then there exists a symmetric signature $h\in\gc(\{f\})$ of GHZ type.
    \label{lem:Wsym}
\end{lemma}

\begin{lemma}[{\cite{backens2017holant+}}]
Suppose there is an irreducible ternary signature $f\in K\mathcal{M}$ and a binary signature $g\notin \langle K\mathcal{M}\rangle$. Then there exists a symmetric signature $h\in\gc(\{f,g\})$ of GHZ type. The same statement also holds after replacing $K\mathcal{M}$ with $KX\mathcal{M}$. 

    \label{lem:Vasym}
\end{lemma}

In the following analysis, we assume there is an irreducible ternary signature $f\in K\mathcal{M}\cap\mathcal{F}$, and the analysis when $f\in KX\mathcal{M}$ is similar. Furthermore, we consider this case in the setting of $\khol$. We remark that $K^{-1}\pin_0=[1,1]$, and by Corollary \ref{cor:hw> construct Delta1} and Theorem \ref{thm:decomposition lemma}, we have $\khol(\widehat{\mathcal{F}},\pin_0)\le_T\khol(\widehat{\mathcal{F}})$ since $\widehat{f}\in \mathcal{M}\cap \widehat{\mathcal{F}}$.
We begin with the analysis of $\widehat{f}$.
\begin{lemma}\label{lem:no up triangle}
     Suppose $\widehat{g}\in \mathcal{F}$ satisfies $M_{\widehat{g}}=\begin{pmatrix}
        a & b\\
        c & 0
    \end{pmatrix}, abc\neq0$. Then,
    $$\khol(\widehat{\mathcal{F}},[1,1])\equiv_T  \hol^c(\mathcal{F})$$
\end{lemma}
\begin{proof}
     By mating the second variable of two $\widehat{g}$ via $\neq_2$, we realize $\widehat{h}$ satisfying $M_{\widehat{h}}=\begin{pmatrix}
        2ab & bc\\
        bc & 0
    \end{pmatrix}$. Since $abc\neq 0$, let $\lambda=2ab/bc\neq 0$ and $M_{\widehat{h}}$ is a multiple of $A=\begin{pmatrix}
        \lambda & 1\\
        1 & 0
    \end{pmatrix}$. Consequently $AX=\begin{pmatrix}
        1 & \lambda\\
        0 & 1
    \end{pmatrix}$ and $(AX)^k\begin{pmatrix}
        1 \\
        1 
    \end{pmatrix}=\begin{pmatrix}
        k\lambda+1 \\
        1 
    \end{pmatrix}$. This means we can realize $[k\lambda+1,1]$ using $\widehat{h}$ and $[1,1]$ in the setting of \khol. Since $\lambda\neq 0$, by Lemma \ref{lem:interpolation} and Theorem \ref{thm:holographic transformation equivalence} we have
$
    \khol(\widehat{\mathcal{F}},[1,1])\equiv_T  \khol(\widehat{\mathcal{F}},[1,1],[1,-1])\equiv_T  \hol^c(\mathcal{F})
$.
\end{proof}
\begin{lemma} \label{lem:only-1,1,0,0}
    Suppose $\widehat{f}\in \mathcal{M}\cap\widehat{\mathcal{F}}$ is an irreducible ternary function and not a multiple of $[-1,1,0,0]$. Then,
    $$\hol(\neq_2\mid \widehat{\mathcal{F}},[1,1])\equiv_T  \hol^c(\mathcal{F})$$
\end{lemma}

\begin{proof}
    Let $\widehat{f}(100)=a,\widehat{f}(010)=b,\widehat{f}(001)=c,\widehat{f}(000)=d$. Since $\widehat{f}$ is irreducible, $abc\neq0$. Since $\widehat{f}$ is not a multiple of $[-1,1,0,0]$, without loss of generality we may assume $a+d\neq0$. By connecting $[1,1]$ to the first variable of $\widehat{f}$ via $\neq_2$, we realize a binary signature $\widehat{g}$ satisfying $M_{\widehat{g}}=\begin{pmatrix}
        a+d & b\\
        c & 0
    \end{pmatrix}$. Since $a+d\neq0$, the proof is completed by Lemma \ref{lem:no up triangle}.
\end{proof}

By Lemma \ref{lem:only-1,1,0,0} and Theorem \ref{thm:holantc dichotomy}, the only unsolved case is when $\widehat{f}=[-1,1,0,0]$. Consequently we can assume $[-1,1,0,0]\in \widehat{\mathcal{F}}$. In this case, firstly we consider the form of unary and binary signatures in $\widehat{\mathcal{F}}$.

\begin{lemma}
    Suppose $[-1,1,0,0],\widehat{g}\in \widehat{\mathcal{F}}$,  where $\widehat{g}=[a,1],a\neq 1$. Then 
    $$\hol(\neq_2\mid \widehat{\mathcal{F}},[1,1])\equiv_T  \hol^c(\mathcal{F})$$
    \label{lem:Rnounary}
\end{lemma}

\begin{proof}
   By connecting one copy of $[a,1]$ to $[-1,1,0,0]$ via $\neq_2$, we realize a binary signature $\widehat{g}$ satisfying $M_{\widehat{g}}=\begin{pmatrix}
        a-1 & 1\\
        1 & 0
    \end{pmatrix}$. Since $a-1\neq0$, the proof is completed by Lemma \ref{lem:no up triangle}.
\end{proof}

\begin{lemma}
    Suppose $[-1,1,0,0],\widehat{g}\in \widehat{\mathcal{F}}$, where $\widehat{g}$ is an irreducible binary signature and not a multiple of $[0,1,0]$. Then one of the following holds:
    \begin{enumerate}
        \item There exists an invertible matrix $N$ such that 
        
$\hol(\mathcal{F},\pin_0)\equiv_T\ccsp(=_2N,\pin_0N, N^{-1}\mathcal{F}, N^{-1}\pin_0)$.
        \item $\khol(\widehat{\mathcal{F}},[1,1])\equiv_T \hol^c(\mathcal{F})$.
    \end{enumerate}
     \label{lem:Rnobinary}
\end{lemma}

\begin{proof}
    Let $M_{\widehat{g}}=\begin{pmatrix}
        a&b \\
        c&d
    \end{pmatrix}$. If $d\neq 0$, $\widehat{g}\notin\langle\mathcal{M}\rangle$, and since $[-1,1,0,0]\in\mathcal{M}$, we can apply Lemma \ref{lem:Vasym} then Lemma \ref{lem:symGHZto=3}, and the first statement holds.

    Otherwise, $d=0$. Since $\widehat{g}$ is irreducible, $bc\neq0$. In this case, if $a\neq0$, the second statement holds by Lemma \ref{lem:no up triangle}. Otherwise, $a=0$ and $b\neq c$, since $\widehat{g}$ is not a multiple of $[0,1,0]$. By connecting a $[1,1]$ to the second variable of $\widehat{g}$ via $\neq_2$, we realize $[b/c,1],b/c\neq 1$ up to a constant factor. Consequently Lemma \ref{lem:Rnounary} can be applied and the second statement holds.
\end{proof}

This result can also be extended to irreducible ternary signatures.
\begin{lemma}
    Suppose $[-1,1,0,0],\widehat{g}\in \widehat{\mathcal{F}}$, where $\widehat{g}$ is an irreducible ternary signature and not a multiple of $[-1,1,0,0]$. Then one of the following holds:
    \begin{itemize}
        \item There exists an invertible matrix $N$ such that 
$$\hol(\mathcal{F},\pin_0)\equiv_T\ccsp(=_2N,\pin_0N, N^{-1}\mathcal{F}, N^{-1}\pin_0).$$
        \item $\khol(\widehat{\mathcal{F}},[1,1])\equiv_T \hol^c(\mathcal{F})$.
    \end{itemize}
     \label{lem:Rnoternary}
\end{lemma}

\begin{proof}
 We prove this theorem in the setting of \khol. By Corollary \ref{cor:hw> construct Delta1} and Theorem \ref{thm:decomposition lemma}, we may realize $\pin_0$. By connecting $\pin_0,[1,1]$ to $\widehat{g}$ via $\neq_2$, and making a self-loop on $\widehat{g}$ via $\neq_2$, several unary and binary signatures can be realized. When $\widehat{g}$ is an irreducible ternary signature and not a multiple of $[-1,1,0,0]$, it can always be verified that at least one of these signatures does not belong to $\langle\{\pin_0,[1,1],[0,1,0]\}\rangle$. Then we are done by Lemma \ref{lem:Rnounary} or Lemma \ref{lem:Rnobinary}. We present the detailed analysis as follows.

    Suppose $\widehat{g}$ has the following form.   
    \begin{table}[h]
\begin{tabular}{l|llllllll}
input  & 000 & 001 & 010 & 100 & 011 & 101 & 110 & 111 \\ \hline
output & $s$ & $f$ & $e$ & $d$ & $c$ & $b$ & $a$ & $t$  
\end{tabular}
\end{table}

We list out some common gadgets used in this proof. By connecting two copies of $\pin_0$ to $\widehat{g}$ via $\neq_2$, we can realize $[a,t],[b,t],[c,t]$. By connecting a $\pin_0$ and a $[1,1]$ to $\widehat{g}$ via $\neq_2$, we can realize $[a+d,b+t],[b+d,a+t],[a+e,c+t],[c+e,a+t],[b+f,c+t],[c+f,b+t]$. By adding a self-loop on $\widehat{g}$ via $\neq_2$, we can realize $[e+f,a+b],[d+f,a+c],[e+d,b+c]$. By connecting a $\pin_0$ to $\widehat{g}$ via $\neq_2$, we can realize $\begin{pmatrix}
    d& a\\
    b& t
\end{pmatrix},\begin{pmatrix}
    e& a\\
    c& t
\end{pmatrix},\begin{pmatrix}
    f& c\\
    b& t
\end{pmatrix}$. 

\textbf{Case 1:} If $t\neq 0$, without loss of generality we may assume $t=1$. Hence we can realize $[a,1],[b,1],[c,1]$ using double $\pin_0$. Then by Lemma \ref{lem:Rnounary}, we are done unless $a=b=c=1$. Now suppose $a=b=c=t=1$, and we can realize $[d,1,1],[e,1,1],[f,1,1]$ using a single $\pin_0$. By Lemma \ref{lem:Rnobinary}, we are done unless $d=e=f=1$. Now we further suppose $d=e=f=1$. By connecting double $[1,1]$ to $\widehat{g}$ via $\neq_2$, we can realize $[s+3,4]$.  By Lemma \ref{lem:Rnounary}, we are done unless $s=1$. If $s=1$, the signature is exactly $[1,1]^{\otimes 3}$, a contradiction to the fact that $\widehat{g}$ is irreducible. 

\textbf{Case 2:}  Now suppose $t=0$. 

\textbf{Case 2.1:} If $a=b=c=0$ does not hold, without loss of generality we assume $a\neq 0$, and further assume $a=1$. Hence we can realize  $\begin{pmatrix}
    d& 1\\
    b& 0
\end{pmatrix}$ using a single $\pin_0$. By Lemma \ref{lem:Rnobinary}, we are done unless $bd=0$. Using similar analysis to $\begin{pmatrix}
    e& a\\
    c& t
\end{pmatrix},\begin{pmatrix}
    f& c\\
    b& t
\end{pmatrix}$, we may further suppose $ce=0,bcf=0$.

\textbf{Case 2.1.1:} If $bc\neq0$, then $d=e=f=0$. By self-loop we realize $[0,a+b],[0,a+c],[0,b+c]$ and at least one of them is non-trivial since $a\neq 0$. Consequently we are done by Lemma \ref{lem:Rnounary}. 

\textbf{Case 2.1.2:} If $b=c=0$, we can realize $[b+d,a+t]=[d,1]$ and $[c+e,a+t]=[e,1]$. By Lemma \ref{lem:Rnounary}, we are done unless $d=e=1$. If further $d=e=1$, now we can realize $[e+f,a+b]=[1+f,1]$ and we are done unless $f=0$. However, when $t=b=c=f=0$, $\widehat{g}(x_1,x_2,x_3)=\widehat{g'}(x_1,x_2)\otimes \pin_0(x_3)$, a contradiction.

\textbf{Case 2.1.3:} Without loss of generality we assume $b\neq 0,c=0$. We also have $d=0$ by $bd=0$. We can realize $[b+d,a+t]=[b,1],[c+e,a+t]=[e,1],[c+f,b+t]=[f,b]$ and by Lemma \ref{lem:Rnounary} we may assume $b=e=f=1$. By connecting a single $[1,1]$ we can realize $\begin{pmatrix}
    s+d& a+e\\
    b+f& c+t
\end{pmatrix}=\begin{pmatrix}
    s& 2\\
    2& 0
\end{pmatrix}$ and by Lemma \ref{lem:Rnobinary} we may assume $s=0$. Now we have $a=b=e=f=1$ and $c=d=s=t=0$, and we can write $\widehat{g}(x_1,x_2,x_3)=[1,1](x_1)\otimes\neq_2(x_2,x_3)$, a contradiction.

\textbf{Case 2.2:} $a=b=c=0$. Since $\widehat{g}$ is irreducible, we have $d,e,f\neq 0$. By connecting a single $[1,1]$ to the first variable we can realize $\begin{pmatrix}
    s+d& a+e\\
    b+f& c+t
\end{pmatrix}=\begin{pmatrix}
    s+d& e\\
    f& 0
\end{pmatrix}$ and by Lemma \ref{lem:Rnobinary} we may assume $d=-s$. Similarly by connecting to other variables, we may assume $e=-s,f=-s$. Consequently, $\widehat{g}$ is a multiple of $[-1,1,0,0]$, a contradiction.
\end{proof}

By Lemma \ref{lem:Rnounary}, \ref{lem:Rnobinary} and \ref{lem:Rnoternary}, the only case left is that all irreducible unary, binary and ternary signatures in $\widehat{\mathcal{F}}$ belong to $\{\Delta_0,[1,1],[0,1,0],[-1,1,0,0]\}$. The following lemma analyzes this last case.

\begin{lemma}\label{lem:yi dingyou unary or binary or tenary}
    Suppose $[-1,1,0,0]\in\widehat{\mathcal F}$ and $\mathcal{R}=\{\pin_0\}\cup\{[-k,1,0,\ldots,0]_{k+2}\mid k\geq -1\}$. Then one of the following statements holds:
    \begin{enumerate}
        \item $\widehat{\mathcal F}\subseteq\langle\mathcal{R}\rangle$.
        \item An irreducible signature $\widehat{h}\notin\mathcal{R}$ with arity less than 4 can be realized in the setting of $\khol(\widehat{\mathcal F},[1,1])$.
    \end{enumerate}
\end{lemma}
\begin{proof}

    It can be verified that $[-k,1,0,\ldots,0]_{k+2}=2\times[1,0]^{\otimes(k+2)}+\textbf{sym}([1,0]^{\otimes(k+1)},[-1,1])$. The symbol $\textbf{sym}([1,0]^{\otimes(k+1)},[-1,1])$ means the sum of $k+2$ distinct signatures, each of which has the same tensor factors, that is $k+1$ copies of $[1,0]$ and a $[-1,1]$, and $[-1,1]$ corresponds to the $i$th variable in the $i$th signature. For example, $\textbf{sym}([1,0]^{\otimes 2},[-1,1])=[-1,1]\otimes[1,0]\otimes[1,0]+[1,0]\otimes[-1,1]\otimes[1,0]+[1,0]\otimes[1,0]\otimes[-1,1]$. 
    
    By Corollary \ref{cor:hw> construct Delta1} $\pin_0$ can be realized from $[-1,1,0,0]$ in the setting of $\khol(\widehat{\mathcal F},[1,1])$. We now apply holographic transformation by $N=\begin{pmatrix}
    1 & 1\\
    0 & 1
\end{pmatrix}$ on $\khol(\widehat{\mathcal{F}},[1,1])$, resulting in $\hol(\neq_2N^{-1}\mid N\widehat{\mathcal{F}},[2,1])$. It can be directly verified that $N\mathcal{R}=\{\pin_0\}\cup\{[2,1,0,\ldots,0]_{k}\mid k\geq1\}$. If we connect $[2,1]$ or $[1,0]$ to $\neq_2N^{-1}$, we realize $(2,1)N^{-\T}XN^{-1}=(1,0)$ and $(1,0)N^{-\T}XN^{-1}=(0,1)$ on the LHS. In other words, after the holographic transformation, we can use $\pin_0$ and $\pin_1$ to operate on signatures in $N\widehat{\mathcal{F}}$.

We also introduce an important property, namely the inheritance property: given a signature $f$ in $\langle N\mathcal{R}\rangle$ and $\alpha\in\su(f)$, if we change a $0$ in $\alpha$ to 1, resulting in $\alpha^+$, then either $f(\alpha^+)=\frac{1}{2}f(\alpha)$, or $f(\alpha^+)=0$; if we change a $1$ in $\alpha$ to 0, resulting in $\alpha^-$, then $f(\alpha^-)=2f(\alpha)$. By Lemma \ref{lem:upf} and the form of $f$, this property can be directly verified.

We now prove the following statement: given any $f$ of arity greater than 3, if for any $i$, $f^{x_i=0}$ and $f^{x_i=1}$ are in $\langle N\mathcal{R}\rangle$, then $f\in\langle N\mathcal{R}\rangle$. Suppose $f$ is of arity greater than 3, $f^{x_i=0},f^{x_i=1}\in\langle N\mathcal{R}\rangle$, and $f^{x_1=0}=f_1\otimes\ldots\otimes f_s$, where $f_i\in N\mathcal{R}$, $1\leq i\leq s$. Then there are several cases as follows.

\begin{enumerate}
    \item There exists an $i$ such that $f_i=[1,0]$.

    Without loss of generality, we can assume that $f_i$ corresponds to $x_2$. Then $f^{x_1=0,x_2=1}\equiv0$. By the inheritance property, we also have $f^{x_1=1,x_2=1}\equiv0$. Therefore, $f=f^{x_2=0}\otimes[1,0](x_2)\in\langle N\mathcal{R}\rangle.$ 

    \item $s=1$ and $f_1\in N\mathcal{R}-\{\pin_0\}$.

    Without loss of generality we may suppose $\ari(f_1)=d$, 
$f_1=[4v,2v,0,0,\dots,0]_d$, where $v\neq 0$. 
    Since $f^{x_2=0}\in\langle N\mathcal{R}\rangle$, by the inheritance property and the fact that $f^{x_1=0,x_2=0}(0_{d-1})=4v$ we have $f^{x_1=1,x_2=0}(0_{d-1})=2v$ or $0$. Furthermore, the output of $f^{x_1=1}$ is $0$ when the Hamming weight of the input is greater than $1$. Since $f^{x_1=1}\in\langle N\mathcal{R}\rangle$, there are several cases.
    \begin{enumerate}
        \item $f^{x_1=1,x_2=0}(0_{d-1})=0$. By the inheritance property we have $f^{x_1=1}\equiv0$, which implies $f=[1,0](x_1)\otimes f_1\in\langle N\mathcal{R}\rangle.$

        \item $f^{x_1=1,x_2=0}(0_{d-1})=2v$ and the values of $f^{x_1=1}=[2v,v,0,0,\dots,0]_d$. Then $f=\frac{1}{2}[2,1](x_1)\otimes f_1\in\langle N\mathcal{R}\rangle.$

        \item $f^{x_1=1,x_2=0}(0_{d-1})=2v$ and the values of $f^{x_1=1}=[2v,0,0,\dots,0]_d$. Then $f$ is exactly $2v[2,1,0,\ldots,0]_{d+1}\in N\mathcal{R}$.

        \item $f^{x_1=1,x_2=0}(0_{d-1})=2v$ and without loss of generality we can assume that $f(1100_{d-2})=v$ and $f(1010_{d-2})=0$. Then let $g=f^{x_4=0,\ldots,x_d=0}$, and $M_{x_1,x_2x_3}(g)=\begin{pmatrix}
        4v & 2v & 2v & 0 \\
        2v & 0  & v & 0
    \end{pmatrix}$. We have $g\notin\langle N\mathcal{R}\rangle$, which is a contradiction. 
    \end{enumerate}

    \item $s>1$ and each $f_i\in N\mathcal{R}-\{\Delta_0\}$. Noticing that by the inheritance property and the fact that $f^{x_2=0},f^{x_2=1}\in\langle N\mathcal{R}\rangle$, we have that $\su(f^{x_1=1,x_2=0})\subseteq\su(f^{x_1=0,x_2=0})$ and $\su(f^{x_1=1,x_2=1})\subseteq\su(f^{x_1=0,x_2=1})$, and thus we have $\su(f^{x_1=1})\subseteq\su(f^{x_1=0})$. 
    We write $f^{x_1=1}$ as $f_1'\otimes\ldots f_t'$, where for each $i$, $f_i'$ is irreducible. 
    
    We consider the partition of variables from $\var(f_1)$ in $f^{x_1=1}$. If $y,z\in\var(f_1)$ and $y\in\var(f_i')$, $z\in\var(f_j')$, where $i\neq j$, then at least one of $f_i'$ and $f_j'$ is $[1,0]$. Suppose $f_i'=[1,0](y)$. If $f_j'\neq[1,0]$, by pinning all variables other than those in $\{x_1,y,z\}$ to 0, the resulting signature does not belong to $\langle N\mathcal{R}\rangle$ (this signature is isomorphic to the $g$ in Case 2d), which causes a contradiction. Consequently, we have that when $x_1=1$, each $\var(f_i)$ is either corresponding to $[1,0]^{\otimes\ari(f_i)}$, or a part of variables of some $f_j'\in N\mathcal{R}-\{[1,0]\}$. 
    
    If $\var(f_k')$ contains $\var(f_i)$ and $\var(f_j)$ with $i\neq j$, without loss of generality suppose the first variable of $f_i$ and $f_j$ are $x_2,x_3$ respectively. By pinning all variables other than $x_1,x_2,x_3$ to 0, we obtain a ternary signature $g$, which satisfies that $M_{x_1,x_2x_3}(g)=\begin{pmatrix}
        4v & 2v & 2v & v \\
        2v & v  & v & 0
    \end{pmatrix}$. Equivalently, $g=[4v,2v,v,0]$, which does not belong to $\langle N\mathcal{R}\rangle$ and causes a contradiction. In summary, the tensor decomposition of $f^{x_1=1}$ can only be a refinement of that of $f^{x_1=0}$. 
    
    We write $f^{x_1=1}=f_1''\otimes\ldots\otimes f_s''$, where $f_i''$ is not necessary to be irreducible and $\var(f_i)=\var(f_i'')$ for each $i$. By pinning all variables other than those in $\{x_1\}\cup\var(f_i)$ to 0, the analysis for $f^{x_1=0}$ and $f^{x_1=1}$ in Case 2 can be similarly applied to $f_i$ and $f_i''$. If $f_i''$ satisfies the condition in Case 2d, we can similarly find a contradiction. Consequently, we have that each $f_i''$ satisfies one of the Case 2a, 2b and 2c. There are several situations.
    
    \begin{enumerate}
        \item There exists an $i$ such that $f_i''\equiv0$. Then $f^{x_1=1}\equiv0$ and $f=[1,0](x_1)\otimes f^{x_1=0}\in\langle N\mathcal{R}\rangle$. In the following, we suppose $f_i''\not\equiv0$ for all $1\leq i\leq s$.

        \item  Each $f_i''$ satisfies Case 2b, then $f=[2,1](x_1)\otimes f^{x_1=1}\in\langle N\mathcal{R}\rangle$.

        \item
        There are $i\neq j$ such that $f_i''$ and $f_j''$ satisfy Case 2c, without loss of generality suppose the first variable of $f_i''$ and $f_j''$ are $x_2,x_3$ respectively. By pinning all variables other than $x_1,x_2,x_3$ to 0, we obtain a ternary signature $g$, which 
        satisfies that $M_{x_1,x_2x_3}(g)=\begin{pmatrix}
        4v & 2v & 2v & v \\
        2v & 0  & 0 & 0
    \end{pmatrix}$. It can be verified that $g$ is isomorphic to that in Case 2d and consequently $g\notin\langle N\mathcal{R}\rangle$, which is a contradiction.

        \item
        If there is an $i$ such that $f_i''$ satisfies Case 2c while the others satisfy Case 2b, we can assume $i=1$ without loss of generality. Furthermore, we assume that $\ari(f_1)=d_1$ and $\var(f_1)=\{x_2,\dots,x_{d_1+1}\}$. Then it can be directly verified that $f=[2,1,0,\ldots,0]_{d_1+1}(x_1,x_2,\ldots,x_{d_1+1})\otimes f_2\otimes\ldots \otimes f_s\in\langle N\mathcal{R}\rangle$.
    \end{enumerate}
\end{enumerate}

In summary, the statement is proved. By the statement, if $f\notin\langle N\mathcal{R}\rangle$ and $\ari(f)\ge 4$, then there exist $1\le i\le \ari(f)$ and $c\in\{0,1\}$ such that $f^{x_i=c}\notin \langle N\mathcal{R}\rangle$. By applying this statement successively, we can obtain a signature $g\notin N\mathcal{R}$ with arity less than 4. If $\widehat{F}\subseteq\langle \mathcal{R}\rangle$ does not hold, there exists $\widehat{f}\in\widehat{\mathcal{F}}$ such that $N\widehat{f}\notin \langle N\mathcal{R}\rangle$, and by the above analysis a signature $N\widehat{g}\notin \langle N\mathcal{R}\rangle$ with arity less than 4 can be realized on RHS in the setting of $\hol(\neq_2N^{-1}\mid N\widehat{\mathcal{F}},[2,1])$, or equivalently $\widehat{g}\notin \langle \mathcal{R}\rangle$ with arity less than 4 can be realized in the setting of  $\khol(\widehat{\mathcal{F}},[1,1])$. As $[1,1]$ is neither a $\eog$ nor a $\eol$ signature, by Theorem \ref{thm:decomposition lemma}, we may further obtain an irreducible signature $\widehat{h}\notin\mathcal{R}$ with arity less than 4 from $\widehat{g}$.
\end{proof}

 It is noteworthy that when $\widehat{\mathcal{F}}\subseteq\langle\mathcal{R}\rangle\subseteq\langle \mathcal{M}\rangle$, $\khol(\widehat{\mathcal{F}})$ is polynomial time computable. When the second statement holds, by Lemma \ref{lem:Rnounary}, \ref{lem:Rnobinary} and \ref{lem:Rnoternary}, the complexity classification for $\hol(\mathcal{F})$ is done.

\subsection{Other cases}\label{sec:khol}

In this section, we classify the complexity of Case 2-4 in Lemma \ref{lem:1fenlei}. The reductions in this section are summarized as a map, presented as Figure \ref{fig:hol odd map}. We commence with the following lemmas employed in Case 4.

\begin{lemma}[{\cite[Lemma 5.2]{cai2020holantoddarity}}] \label{lem:khol to ccspk}
     For any $k\ge 3$, $\ccsp_k(\neq_2,F) \le_T \khol(=_k,F)$.
\end{lemma}
\begin{lemma} \label{lem: you duoyuan xiangdeng hardness}
    For arbitrary integer $k\ge 3$ and $ab\neq 0$, there exists $\widehat{Q}\in\mathbb{C}^{2\times2}$ such that $\ccsp_k(\neq_2,\widehat{Q}\widehat{\mathcal{F}})\le_T\khol(\widehat{\mathcal{F}},[a,0,\dots,0,b]_{k})$. 
\end{lemma}
\begin{proof}
    Let $q^{2k}=\frac{a}{b}$ and $\widehat{Q}=\begin{pmatrix}
        1/q & 0\\
        0 & q
    \end{pmatrix}$. We have $\widehat{Q}[a,0,\dots,0,b]_{k}=[a(1/q)^k,0,\dots,0,bq^k]_{k}=bq^k[1,0,\dots,0,1]_{k}$. Then $\khol(\widehat{Q}\widehat{\mathcal{F}},[1,0,\dots,0,1]_{k})\equiv_T\khol(\widehat{\mathcal{F}},[a,0,\dots,0,b]_{k})$ by Lemma \ref{lem:quanxi}.
    By Lemma \ref{lem:khol to ccspk}, $\ccsp_k(\neq_2,\widehat{Q}\widehat{\mathcal{F}})\le_T\hol(\neq_2\mid\widehat{Q}\widehat{\mathcal{F}},=_{k})$. The proof is  completed.
\end{proof}

\begin{lemma}\label{lem:youduoyuanxiangdeng suanfa}
    If $\ccsp_k(\neq_2,\widehat{Q}\widehat{\mathcal{F}})$ is tractable, then $\hol(\neq_2\mid\widehat{\mathcal{F}},[a,0,\dots,0,b]_{k})$ is also tractable.
\end{lemma}

Lemma \ref{lem:youduoyuanxiangdeng suanfa} can be proved by directly checking the tractable cases of $\ccsp_k$ problems and holographic transformation. Combining Lemma \ref{lem: you duoyuan xiangdeng hardness} and \ref{lem:youduoyuanxiangdeng suanfa} we have $\ccsp_k(\neq_2,\widehat{Q}\widehat{\mathcal{F}})\equiv_T\khol(\widehat{\mathcal{F}},[a,0,\dots,0,b]_{k})$, which completes the complexity classification in Case 4.

For Case 2 and 3 in Lemma \ref{lem:1fenlei}, we have the following lemma.
\begin{lemma}\label{lem:odd case2 case3 fenlei}
    Let $c\in \{0,1\}$. Then one of the following statements holds:
    \begin{enumerate}
        \item All signatures in $\widehat{\mathcal{F}}\cup\{\pin_c\}$ are $\eog$ (or $\eol$ respectively);
        \item $\hol(Q\mathcal{F},\pin_0, [1,(-1)^c\ii])\equiv_T \khol(\widehat{\mathcal{F}},\pin_c)$ for some $Q\in\mathscr{O}$;
        \item $\khol^c(\widehat{\mathcal{F}})\equiv_T \khol(\widehat{\mathcal{F}},\pin_c)$.
    \end{enumerate}
\end{lemma}
\begin{proof}
    Suppose $c=0$. When $c=1$, the proof is similar.

    If for each $\widehat{f}\in\widehat{\mathcal{F}}$, $\su(\widehat{f})\subseteq \eol$, then the first statement holds. 
    Now suppose there exist some $\widehat{f}\in\widehat{\mathcal{F}}$ and $\alpha\in\su(\widehat{f})$ such that $\alpha\in\eosg$. By applying Lemma \ref{lem:self-loop to pure1 string} we can obtain a signature $\widehat{g}$ of arity $k$ satisfying $\widehat{g}(\textbf{1})=1$ up to a constant factor $\widehat{f}(\alpha)$. 
    
    By connecting $k-1$ copies of $\pin_0$ to $\widehat{g}$ via $\neq_2$, we obtain the signature $(a,1)$. If $a=0$, then the third statement holds. Otherwise, we let $q^2=a$ and $\widehat{Q}=\begin{pmatrix}
        1/q & 0\\
        0 & q
    \end{pmatrix}$. Then $\widehat{Q}[a,1]=q[1,1]$ and $\hol((K\widehat{Q}K^{-1})\mathcal{F},[1,\ii],\pin_0)\equiv_T \khol(\widehat{Q}\widehat{\mathcal{F}},\pin_0,[1,1])\equiv_T \khol(\widehat{\mathcal{F}},\pin_0,[a,1])\equiv_T \khol(\widehat{\mathcal{F}},\pin_0)$. By taking $Q=K\widehat{Q}K^{-1}$, the second statement holds.
\end{proof}

The complexity of the first and second situation in Lemma \ref{lem:odd case2 case3 fenlei} is already classified in Theorem \ref{cor:eol eog dichotomy} and Lemma \ref{lem:hol0 dichotomy}. We now focus on the situation that the third statement holds.

\begin{lemma}
    Suppose $\widehat{f}$ is of arity $k\ge 3$, $\alpha,\beta\in\su(\widehat{f})$ and one of $\alpha,\beta$ is not in $\{0_k,1_k\}$. Then there exists a signature $\widehat{f'}$ satisfying the following conditions.
    \begin{itemize}
        \item $\khol^c(\widehat{f},\widehat{f'})\equiv_T\khol^c(\widehat{f})$ and arity$(\widehat{f'})<k$.
        \item $\alpha',\beta'\in \su(\widehat{f'})$ and $\#_1(\alpha')-\#_1(\beta')=\#_1(\alpha)-\#_1(\beta)$.
    \end{itemize}
    \label{lem:jiangyuan}
\end{lemma}
\begin{proof}
    If $\alpha,\beta$ are the same at the $i$th bit, then by pinning the $i$th bit to $\alpha_i$ we obtain the desired $\widehat{f'}$.
    Now suppose $\alpha,\beta$ are distinct at every bit, which means $\beta=\overline{\alpha}$. As one of $\alpha,\beta$ is not from $\{0_k,1_k\}$, we may assume $\alpha=01\alpha',\beta=10\beta'$ without loss of generality. If $01\beta'\in\su(\widehat{f})$, by pinning $x_1$ to 0 and $x_2$ to 1 we obtain the desired $\widehat{f'}$. Similarly if $10\alpha'\in\su(\widehat{f})$, by pinning $x_1$ to 1 and $x_2$ to 0 we obtain the desired $\widehat{f'}$. Otherwise, we are done by adding a self-loop by $\neq_2$ on $x_1$ and $x_2$.
\end{proof}

\begin{lemma}\label{lem:kholantc fenlei}
    One of the following statements holds for $\khol^c(\widehat{\mathcal{F}})$.
    \begin{enumerate}
        \item All signatures in $\widehat{\mathcal{F}}$ are single-weighted.
        \item There exists a $\widehat{Q}=\begin{pmatrix}
        1/r & 0\\
        0 & r
    \end{pmatrix}$ such that $\khol^c(\widehat{Q}\widehat{\mathcal{F}},[1,0,\dots,0,1]_k)\equiv_T\khol^c(\widehat{\mathcal{F}}).$
    \end{enumerate}
\end{lemma}
\begin{proof}  
    If the first statement does not hold, there exists a signature $\widehat{f}\in \widehat{\mathcal{F}}$ of arity $k$ such that $\alpha,\beta\in\su(\widehat{f})$ and $\#_1(\alpha)<\#_1(\beta)$. Furthermore, we may assume for each $\gamma\in\su(\widehat{f})$, either $\#_1(\gamma)\le \#_1(\alpha)$ or $\#_1(\gamma)\ge \#_1(\beta)$ holds. Now we prove that a signature of the form $[a,0,\dots,0,b]_{\ge 1},ab\neq0$ can be realized in this case. If $\{\alpha,\beta\}=\{0_k,1_k\}$ we are done. Otherwise, by applying Lemma \ref{lem:jiangyuan} successively, we can also realize $\widehat{f'}=[a,0,\dots,0,b]_{\ge 1}$.

    Now we already realize a signature $[a,0,\ldots,0,b]_k$, $ab\neq0$. Let $r^{2k}=\frac{a}{b}$ and $\widehat{Q}=\begin{pmatrix}
        1/r & 0\\
        0 & r
    \end{pmatrix}$, then $\widehat{Q}[a,0,\dots,0,b]_k=[a/r^k,0,\dots,0,br^k]_k=br^k[1,0,\dots,0,1]_k$. Notice that the holographic transformation by $\widehat{Q}$ does not change $\Delta_0$ and $\Delta_1$. By Lemma \ref{lem:quanxi}, we have $\khol^c(\widehat{Q}\widehat{\mathcal{F}},[1,0,\dots,0,1]_k)\equiv_T\khol^c(\widehat{\mathcal{F}},[a,0,\dots,0,b]_k)\equiv_T\khol^c(\widehat{\mathcal{F}})$.
\end{proof}

In the first statement of Lemma \ref{lem:kholantc fenlei},  the complexity of $\khol^c(\widehat{\mathcal{F}})$ is already classified by Theorem \ref{thm:single weighted dichotomy}.
In the second statement of Lemma \ref{lem:kholantc fenlei}, when $k\geq3$, the complexity is already classified by Lemma \ref{lem: you duoyuan xiangdeng hardness} and \ref{lem:youduoyuanxiangdeng suanfa}. When $k=1$, the complexity is already classified by Lemma \ref{lem:hol0 dichotomy}. By the following lemma and Theorem \ref{thm:holantc dichotomy} we complete the complexity classification when $k=2$.

\begin{lemma}\label{lem:kholantc to holantc with neq}
$\hol^c(\widehat{\mathcal{F}},\neq_2)\equiv_T\khol^c(\widehat{\mathcal{F}},=_2)$.   
\end{lemma}
\begin{proof}
    By connecting $[1,0,1]$ to 2 copies of $[0,1,0]$, we realize $[1,0,1]$. By connecting $[0,1,0]$ to 2 copies of $[1,0,1]$, we realize $[0,1,0]$. Then we have:

    \begin{align*}
        \hol^c(\widehat{\mathcal{F}},\neq_2)&\equiv_T\hol(=_2\mid \widehat{\mathcal{F}},\neq_2,\pin_0,\pin_1)\\
        &\equiv_T\hol(\neq_2,=_2\mid \widehat{\mathcal{F}},=_2,\neq_2,\pin_0,\pin_1)\\
        &\equiv_T \hol(\neq_2\mid \widehat{\mathcal{F}},=_2,\pin_0,\pin_1)\\
        &\equiv_T\khol^c(\widehat{\mathcal{F}},=_2).
    \end{align*}   
\end{proof}

\subsection{Proof of the main theorem}\label{sec:proof of main}
In preparation for the proof of Theorem \ref{thm:main theorem}, we introduce an important lemma.
\begin{lemma}\label{lem:orthogonal does not change tractable}
    Let $\mathcal{F}$ be a set of complex-valued signatures that satisfies condition ($\mathcal{PC}$). Then for any $O\in\mathscr{O}$, $O\mathcal{F}$ satisfies condition ($\mathcal{PC}$).
\end{lemma}
\begin{proof}
    We will prove that if $\mathcal{F}$ satisfies some case in condition ($\mathcal{PC}$), then $O\mathcal{F}$ satisfies the same case. The Case 3, 5, 6, 7 in condition ($\mathcal{PC}$) can be directly verified by definition. We now consider the Case 1, 2 and 4. We commence with proving two claims as follows.

    The first claim is that for any $O\in\mathscr{O}$, $OK=KD$ or $OK=KXD$, where $D\in\mathbb{C}^{2\times2}$ is a diagonal matrix. Now we prove this claim. Suppose $O=\begin{pmatrix}
        a & b \\
        c & d
    \end{pmatrix}$. Since $O\in\mathscr{O}$, we have $a^2+b^2=c^2+d^2=1$ and $ac+bd=0$. 
    \begin{enumerate}
        \item If $b=0$, then $cd=0$. Since $O\in\mathbf{GL}_2(\mathbb{C})$, we have $c=0$. In this case $O=\begin{pmatrix}
        1 & 0 \\
        0 & 1
    \end{pmatrix}$ or $\begin{pmatrix}
        1 & 0 \\
        0 & -1
    \end{pmatrix}$. The claim is true with $D=\begin{pmatrix}
        1 & 0 \\
        0 & 1
    \end{pmatrix}$. 
    \item If $c=0$, the situation is symmetric to the case $b=0$.
    \item If $a=0$, then $d=0$. Therefore, $O=\begin{pmatrix}
        0 & 1 \\
        1 & 0
    \end{pmatrix}$ or $\begin{pmatrix}
        0 & 1 \\
        -1 & 0
    \end{pmatrix}$. We have $OK=K\begin{pmatrix}
        0 & -2\ii \\
        2\ii & 0
    \end{pmatrix}$ or $K\begin{pmatrix}
        2\ii & 0 \\
        c & -2\ii
    \end{pmatrix}$.
    \item If $d=0$, the situation is symmetric to the case $a=0$.
    \item If $abcd\neq0$, then $(\frac{a}{b})^2=(\frac{d}{c})^2$. Therefore, $b^2=c^2,a^2=d^2$. If $a=d$, then $b=-c$, and we have $O=\begin{pmatrix}
        a & b \\
        -b & a
    \end{pmatrix}$. In this case $OK=K\begin{pmatrix}
        a+b\ii & 0 \\
        0 & a-b\ii
    \end{pmatrix}$. If $a=-d$, then $b=c$, and we have $O=\begin{pmatrix}
        a & b \\
        b & -a
    \end{pmatrix}$. In this case $OK=K\begin{pmatrix}
        0 & a-b\ii \\
        a+b\ii & 0
    \end{pmatrix}$.
    \end{enumerate}


    We now present the second claim: suppose $D=\begin{pmatrix}
        a & 0 \\
        0 & b
    \end{pmatrix}$ is an invertible diagonal matrix and $\widehat{\mathcal{F}}$ satisfies one of the cases among Case 1, 2, 4, then $D\widehat{\mathcal{F}}$ also satisfies the same case. Notice that by Lemma \ref{lem:quanxi}, for any $\widehat{f}\in\widehat{\mathcal{F}}$, $\su(\widehat{f})=\su(D\widehat{f})$. Therefore, the claim holds for Case 4. We also remark that, for any $\widehat{f}$ of arity $k$ and $\alpha\in\su(\widehat{f})$ of Hamming weight $d$, $D\widehat{f}(\alpha)=a^{k-d}b^d\widehat{f}(\alpha)$. Consequently, for signature $\widehat{f}$ of arity $2d$ we have $(D\widehat{f})|_{\eo}=a^{d}b^d\widehat{f}|_{\eo}$. Moreover, for a \sw\  signature $\widehat{f}$ of arity $k$ which takes the value $0$ on all input strings whose Hamming weight is not equal to $d$, $D\widehat{f}=a^{k-d}b^d\widehat{f}$ and consequently $(D\widehat{f})_{\to \eo}=a^{k-d}b^d\widehat{f}_{\to \eo}$. These properties ensure that the claim holds for Case 1 and 2.
    
    Furthermore, the second claim still holds by replacing $D$ with $XD$.  The only difference between $D\widehat{f}$ and $XD\widehat{f}$ is that the symbols of $0$ and $1$ are exchanged. Therefore, $\widehat{f}$ and $XD\widehat{f}$ now satisfy two symmetric conditions within the same case. For example, if $\widehat{f}$ is $\eog$, then $XD\widehat{f}$ is $\eol$.

    Combining the two claims above, the proof is completed.
\end{proof}

We remark that by Lemma \ref{lem:orthogonal does not change tractable} and reduction to absurdity, for any $O\in\mathscr{O}$, if $\mathcal{F}$ does not satisfy condition ($\mathcal{PC}$), then $O\mathcal{F}$ also does not satisfy condition ($\mathcal{PC}$). A direct corollary is presented in the following.

\begin{corollary}
\label{cor:no pc no others}
    Let $\mathcal{F}$ be a set of complex-valued signatures and $\mathcal{F}$ does not satisfy condition ($\mathcal{PC}$). Then for any $O\in\mathscr{O}$, $\ccsp(O\mathcal{F})$, $\ccsp_2(O\mathcal{F})$, $\hol^c(O\mathcal{F})$, $\khol(K^{-1}O\mathcal{F})$ are \#P-hard.  
\end{corollary}

Now we are ready to prove Theorem \ref{thm:main theorem}.

\begin{proof}[Proof of Theorem \ref{thm:main theorem} and Lemma \ref{lem:hol0 dichotomy}]
By the algorithm results in dichotomies for $\#\eo$ (Theorem \ref{thm:eo dichotomy}), $\khol$ defined by single-weighted signatures (Theorem \ref{thm:single weighted dichotomy}), $\ccsp$ (Theorem \ref{thm:CSPdichotomy}), $\ccsp_2$ (Theorem \ref{thm:csp2 dichotomy}), $\hol^c$ (Theorem \ref{thm:holantc dichotomy}) and Theorem \ref{thm:holographic transformation equivalence}, when $\mathcal{F}$ satisfies condition ($\mathcal{PC}$), $\hol(\mathcal{F})$ is computationally tractable to varying degrees.

We now prove the hardness part. We show the sketch of proof in two graphs in Section\ref{main}. Assume $\mathcal{F}$ does not satisfy condition ($\mathcal{PC}$) and contains a non-trivial signature of odd arity. By Lemma \ref{lem:1fenlei} there are four cases.

Case 1 is that after a holographic transformation by some $O\in\mathscr{O}$, the problem is equivalent to $\hol(\mathcal{F},\Delta_0)$. By Lemma \ref{lem:2fenlei} there are three situations.  By Corollary \ref{cor:no pc no others} we have that the third situation is \#P-hard. By Lemma \ref{lem:hol =4 csp2 part1}, \ref{lem:hol =4 csp2 part2} and Corollary \ref{cor:no pc no others}, the second situation is \#P-hard. By Lemma \ref{lem:symmetric GHZ goto CSP} ($\mathcal{F}$ is not $\mathscr{A}$-transformable or $\mathscr{P}$-transformable), \ref{lem:GHZsym}, \ref{lem:Wsym}, \ref{lem:only-1,1,0,0}, \ref{lem:Rnounary}, \ref{lem:Rnobinary}, \ref{lem:Rnoternary}, \ref{lem:yi dingyou unary or binary or tenary} and Corollary \ref{cor:no pc no others}, the first situation is \#P-hard.  By now, we have already proved Lemma \ref{lem:hol0 dichotomy}.

We now consider  Case 4 in Lemma \ref{lem:1fenlei}. By Lemma \ref{lem: you duoyuan xiangdeng hardness} and \ref{lem:youduoyuanxiangdeng suanfa}, it is sufficient for us to classify the complexity of $\ccsp_d(\neq_2,\widehat{Q}\widehat{\mathcal{F}})$, where $\widehat{Q}=\begin{pmatrix}
    1/q & 0 \\
    0 & q
\end{pmatrix}$ and $q\neq0$. Assume it is tractable, by Theorem \ref{thm:cspd neq_2 dichotomy}, $\widehat{Q}K^{-1}\mathcal{F}$ is a subset of $\mathscr{P}$ or $\mathscr{A}_d^r$. The former situation satisfies that $\mathcal{F}$ is $\mathscr{P}$-transformable. Since $\neq_2T_d^r$ is still $\neq_2$ up to a constant, the latter situation satisfies that $\mathcal{F}$ is $\mathscr{A}$-transformable. Both situations lead to a contradiction as $\mathcal{F}$ does not satisfy condition ($\mathcal{PC}$). Therefore, the assumption is false and $\ccsp_d(\neq_2,\widehat{Q}\widehat{\mathcal{F}})$ is \#P-hard.

For Case 2 and 3 in Lemma \ref{lem:1fenlei}, there are three situations by Lemma \ref{lem:odd case2 case3 fenlei}. Since $\mathcal{F}$ does not satisfy condition ($\mathcal{PC}$), the first situation is \#P-hard by Corollary \ref{cor:eol eog dichotomy}. By the first claim in the proof of Lemma \ref{lem:orthogonal does not change tractable}, we have $K\widehat{Q}K^{-1}=OKK^{-1}=O$ for some $O\in\mathscr{O}$. Therefore, the second situation is reduced to Case 1 in Lemma 
\ref{lem:1fenlei}, and it is \#P-hard. The only left situation is $\khol^c(\widehat{\mathcal{F}})$. If all signatures in $\widehat{\mathcal{F}}$ are single-weighted, by Theorem \ref{thm:single weighted dichotomy} it is \#P-hard. By Lemma \ref{lem:kholantc fenlei}, we only need to consider $\khol^c(\widehat{Q}\widehat{\mathcal{F}},[1,0,\ldots,0,1]_k)$, where $\widehat{Q}=\begin{pmatrix}
    1/r & 0 \\
    0 & r
\end{pmatrix}$ with $r\neq0$. If $k\geq 3$, by the analysis of Case 4, it is \#P-hard. If $k=1$, by the first claim in the proof of Lemma \ref{lem:orthogonal does not change tractable}, it is equivalent to $\hol([1,\ii],[1,-\ii],O\mathcal{F},\Delta_0)$ with some $O\in\mathscr{O}$, which is already analyzed in the proof of Case 1. 

When $k=2$, $\khol^c(\widehat{Q}\widehat{\mathcal{F}},=_2)$ is computationally equivalent to $\hol^c(\widehat{Q}\widehat{\mathcal F},\neq_2)$ by Lemma \ref{lem:kholantc to holantc with neq}. Since $\mathcal{F}$ does not satisfy condition ($\mathcal{PC}$), $\widehat{Q}\widehat{\mathcal{F}}\not\subseteq\langle\mathcal{T}\rangle$. We also have $\neq_2\notin\langle K\mathcal{M}\rangle\cup\langle KX\mathcal{M}\rangle\cup\mathscr{L}$. If $\widehat{Q}\widehat{\mathcal{F}}\cup\{\neq_2\}$ is $\mathscr{P}$-transformable, then there exists an $M\in\mathbf{GL}_2(\mathbb{C})$ such that $M\widehat{Q}\widehat{\mathcal{F}}\cup\{M\neq_2\}\subseteq\mathscr{P}$ . By Lemma \ref{lem:dagonggaocheng}, $\neq M^{-1}\in \mathscr{P}$. Therefore, $=_2K\widehat{Q}^{-1}M^{-1}\in\mathscr{P}$ and $\mathcal{F}$ is $\mathscr{P}$-transformable, which is a contradiction. If $\widehat{Q}\widehat{\mathcal{F}}\cup\{\neq_2\}$ is $\mathscr{A}$-transformable, the analysis is similar and  we have $\mathcal{F}$ is $\mathscr{A}$-transformable, which is a contradiction. By Theorem \ref{thm:holantc dichotomy} we have $\hol^c(\widehat{Q}\widehat{\mathcal F},\neq_2)$ is \#P-hard. 

In summary, we have proved that when $\mathcal{F}$ does not satisfy condition ($\mathcal{PC}$), $\hol(\mathcal{F})$ is \#P-hard. Therefore, Theorem \ref{thm:main theorem} is proved.
\end{proof}
\section{Conclusion}\label{ccls}

In this article, we prove a generalized decomposition lemma for complex-valued \hol. Based on this lemma, we further prove a dichotomy for \hol\ when a non-trivial signature of odd arity exists ($\holodd$).

We emphasize that this dichotomy for $\holodd$ is still an $\pnp$ vs. \#P dichotomy due to the dichotomy for \ceo\ in \cite{meng2025fpnp}.  As stated in \cite{meng2025fpnp}, a specific problem defines the NP oracle and leads to the study of Boolean constraint satisfaction problems \cite{feder2006classification}. It is our hope that a complete complexity classification for this kind of problems can be established in the future.

Furthermore, it is worthwhile to pursuit the dichotomy for complex-valued \hol. By Theorem \ref{thm:decomposition lemma} and \ref{thm:main theorem}, the only remaining case is when all signatures in $\mathcal{F}$ are of even arity and irreducible. Nevertheless, the analysis of this case may present significant challenges, as was evidenced in the proof of its sub-cases \cite{shao2020realholant,caifu2023eightvertex}.



\bibliography{ref}

\appendix

\end{document}